\documentclass{article}
\usepackage{iftex}
\ifPDFTeX
\pdfoutput=1
\fi

\usepackage{amsthm}
\usepackage{rohitagr}
\usepackage{url}
\newcommand{\email}[1]{\href{mailto:#1}{\nolinkurl{#1}}}
\usepackage{fullpage}

\usepackage[capitalize]{cleveref}
\usepackage{tabularx}
\usepackage{threeparttable}

\crefformat{footnote}{#2\footnotemark[#1]#3}

\theoremstyle{plain}
\newtheorem{thm}{Theorem}[section]
\Crefname{thm}{Theorem}{Theorems}
\newtheorem{prop}[thm]{Proposition}
\newtheorem{lem}[thm]{Lemma}
\Crefname{lem}{Lemma}{Lemmas}
\newtheorem{cor}[thm]{Corollary}
\Crefname{cor}{Corollary}{Corollaries}

\theoremstyle{definition}
\newtheorem{defn}[thm]{Definition}
\newtheorem{exmp}[thm]{Example}

\theoremstyle{remark}
\newtheorem{rmrk}[thm]{Remark}
\numberwithin{equation}{thm}

\newcommand{\tv}{d_{TV}}
\newcommand{\X}{\mathcal{X}}
\newcommand{\F}{\mathcal{F}}
\newcommand{\cG}{\mathcal{G}}
\newcommand{\cE}{\mathcal{E}}

\newcommand{\R}{\mathbb{R}}

\DeclareMathOperator{\Ext}{Ext}
\DeclareMathOperator{\Disp}{Disp}
\DeclareMathOperator{\Waste}{W}
\DeclareMathOperator{\Samp}{Samp}
\DeclareMathOperator{\Supp}{Supp}

\DeclareMathOperator{\D}{D}
\DeclareMathOperator{\KL}{KL}
\DeclareMathOperator{\poly}{poly}

\DeclareMathOperator{\Ent}{H}
\DeclareMathOperator{\TildeEnt}{\tilde{H}}

\newcommand{\minent}{\Ent_\infty}
\newcommand{\condminent}{\TildeEnt_\infty}

\DeclareMathOperator{\maxdev}{max\,dev}

\newcommand{\zo}{\set{0,1}}
\def\titleeps/{\texorpdfstring{$\eps$}{\ifpdfstringunicode{\unichar{"03B5}}{epsilon}}}
\def\titledelta/{\texorpdfstring{$\delta$}{\ifpdfstringunicode{\unichar{"03B4}}{delta}}}

\title{Samplers and Extractors for Unbounded Functions}
\author{Rohit Agrawal\thanks{John A. Paulson School of Engineering and
Applied Sciences, Harvard University, Cambridge, MA 02138, USA.
Email: \email{rohitagr@seas.harvard.edu}.
Supported by the Department of Defense (DoD) through the National
Defense Science \& Engineering Graduate Fellowship (NDSEG)
Program.}}
\hypersetup{pdfauthor={Rohit Agrawal}}

\begin{document}

\maketitle

\begin{abstract}
	B\l{}asiok (SODA'18) recently introduced the notion of a subgaussian
	sampler, defined as an averaging sampler for approximating the mean of
	functions $f:\zo^m \to \R$ such that $f(U_m)$ has subgaussian tails,
	and asked for explicit constructions. In this work, we give the first
	explicit constructions of subgaussian samplers (and in fact averaging
	samplers for the broader class of subexponential functions) that match
	the best known constructions of averaging samplers for $[0,1]$-bounded
	functions in the regime of parameters where the approximation error
	$\eps$ and failure probability $\delta$ are subconstant. Our
	constructions are established via an extension of the standard notion
	of randomness extractor (Nisan and Zuckerman, JCSS'96) where the error
	is measured by an arbitrary divergence rather than total variation
	distance, and a generalization of Zuckerman's equivalence (Random
	Struct.\ Alg.'97) between extractors and samplers.  We believe that
	the framework we develop, and specifically the notion of an extractor
	for the Kullback--Leibler (KL) divergence, are of independent
	interest. In particular, KL-extractors are stronger than both standard
	extractors and subgaussian samplers, but we show that they exist with
	essentially the same parameters (constructively and
	non-constructively) as standard extractors.
\end{abstract}

\section{Introduction}
\subsection{Averaging samplers}

Averaging (or oblivious) samplers, introduced by Bellare and Rompel
\cite{bel_rom_randomnessefficient_1994}, are one of the main objects of
study in pseudorandomness. Used to approximate the mean of a
$[0,1]$-valued function with minimal randomness and queries, an
averaging sampler takes a short random seed and produces a small set of
correlated points such that any given $[0,1]$-valued function will (with
high probability) take approximately the same mean on these points as on
the entire space. Formally,
\begin{defn}[\cite{bel_rom_randomnessefficient_1994}]\label{intro_defn:sampler}
	A function $\Samp:\zo^n \to \of{\zo^m}^D$ is a \emph{$(\delta,
	\eps)$ averaging sampler} if for all $f:\zo^m \to [0, 1]$,
	it holds that
	\[
		\PR[x\sim U_n]{\abs{\frac1D \sum_{i=1}^Df(\Samp(x)_i) - \E{f(U_m)}} >
		\eps} \leq \delta,
	\]
	where $U_n$ is the uniform distribution on $\zo^n$.  The number
	$n$ is the \emph{randomness complexity} of the sampler, and $D$
	is the \emph{sample complexity}. A sampler is \emph{explicit}
	if $\Samp(x)_i$ can be computed in time $\poly(n,m,\log D)$.
\end{defn}
Traditionally, averaging samplers have been used in the context of
randomness-efficient error reduction for algorithms and protocols, where
the function $f$ is the indicator of a set ($\zo$-valued), or more
generally the acceptance probability of an algorithm or protocol
($[0,1]$-valued).  There has been significant effort in the literature
to establish optimal explicit and non-explicit constructions of
samplers, which we summarize in \cref{table:samplers}. We recommend the
survey of Goldreich \cite{goldreich_sample_2011} for more details,
especially regarding non-averaging
samplers\footnote{\label{foot:nonaveraging}A non-averaging sampler is an
	algorithm $\Samp$ which makes oracle queries to $f$ and outputs an
	estimate of its average which is good with high probability, but need
	not simply output the average of $f$'s values on the queried points.}.

\begin{table}[h!tbp]
	\centering
	\begin{threeparttable}[b]
	\caption{Best known constructions of averaging samplers for $[0,1]$-valued functions}
	\label{table:samplers}
	\begin{tabularx}{\textwidth}{|>{\centering}X|c|c|c|}\hline
		Key Idea & Randomness complexity $n$ &Sample complexity $D$ & Best regime\\\hline
		Pairwise-independent Expander Neighbors \cite{gol_wig_tiny_1997} & $m + O\of{\log(1/\delta)+\log(1/\eps)}$ & $O\of{\frac1{\delta\eps^2}}$  & $\delta = \Omega(1)$\\\hline
		Ramanujan Expander Neighbors\tnote{a}  \cite{kar_pip_sip_timerandomness_1985,gol_wig_tiny_1997} & $m$ & $O\of{\frac1{\delta\eps^2}}$  & $\delta = \Omega(1)$\\\hline
		Extractors \cite{zuckerman_randomnessoptimal_1997,gol_wig_tiny_1997,rei_vad_wig_entropy_conference,gur_uma_vad_unbalanced_2009}
		& $\begin{aligned}[t]&m + (1+\alpha)\cdot \log(1/\delta)\\&\hspace{0.5em}\text{any constant }\alpha>0\end{aligned}$ & $\poly\of{\log(1/\delta),1/\eps}$ & $\eps,\delta=o(1)$\\\hline
	Expander Walk Chernoff \cite{gillman_chernoff_1998}
	& $m + O(\log(1/\delta)/\eps^2)$ & $O\of{\frac{\log(1/\delta)}{\eps^2}}$ & $\eps=\Omega(1)$
	\\\hline
		Pairwise Independence \cite{cho_gol_power_1989} & $O(m)$ & $O\of{\frac1{\delta\eps^2}}$  & None, but simple \\\hline
	Non-Explicit  \cite{zuckerman_randomnessoptimal_1997} & $\begin{aligned}[t]m + \log&(1/\delta)-\log\log(1/\delta)\\ & + O(1)\end{aligned}$ & $O\of{\frac{\log(1/\delta)}{\eps^2}}$ & All\\\hline
		Lower Bound
		\cite{can_eve_gol_lower_1995,zuckerman_randomnessoptimal_1997,rad_ta-_bounds_2000}
									& $\begin{aligned}[t]m &+
									\log(1/\delta)+\log(1/\eps)\\&- \log(D) - O(1)\end{aligned}$ &
										$\Omega\of{\frac{\log{(1/\delta)}}{\eps^2}}$ & N/A\\\hline
	\end{tabularx}
	\begin{tablenotes}
	\item[a] Requires explicit constructions of Ramanujan graphs.
	\end{tablenotes}
	\end{threeparttable}
\end{table}

However, averaging samplers can also have uses beyond bounded functions:
B\l{}asiok \cite{blasiok_optimal_conference}, motivated by an
application in streaming algorithms, introduced the notion of a
\emph{subgaussian sampler}, which he defined as an averaging sampler for
functions $f:\zo^m \to \R$ such that $f(U_m)$ is a subgaussian random
variable. Since subgaussian random variables have strong tail bounds,
subgaussian functions from $\zo^m$ have a range contained in
an interval of length $O(\sqrt{m})$,
and thus one can construct a subgaussian sampler from a $[0,1]$-sampler
by simply scaling the error $\eps$ by a factor of $O(\sqrt m)$.
Unfortunately, looking at \cref{table:samplers} one sees that this
induces a multiplicative dependence on $m$ in the sample complexity, and
for the expander walk sampler induces a dependence of $m\log(1/\delta)$
in the randomness complexity. This loss can be avoided for some
samplers, such as the sampler of Chor and Goldreich
\cite{cho_gol_power_1989} based on pairwise independence (as its
analysis requires only bounded variance) and (as we will show) the
Ramanujan Expander Neighbor sampler of
\cite{kar_pip_sip_timerandomness_1985,gol_wig_tiny_1997}, but B\l{}asiok
showed \cite{blasiok_optimal_private_2018} that the expander-walk
sampler does not in general act as a subgaussian sampler without
reducing the error to $o(1)$. We remark briefly that the
median-of-averages sampler of Bellare, Goldreich, and Goldwasser
\cite{bel_gol_gol_randomness_1993} still works and is optimal up to
constant factors in the subgaussian setting (since the underlying
pairwise independent sampler works), but it is not an averaging
sampler\cref{foot:nonaveraging}, and matching its parameters with an averaging
sampler remains open in general even for $[0,1]$-valued functions.

One of the contributions of this work is to give explicit averaging
samplers for subgaussian functions (in fact even for
\emph{subexponential} functions that satisfy weaker tail bounds)
matching the extractor-based samplers for $[0,1]$-valued functions in
\cref{table:samplers} (up to the hidden polynomial in the sample
complexity). This achieves the best parameters currently known in the
regime of parameters where $\eps$ and $\delta$ are both subconstant, and
in particular has no dependence on $m$ in the sample complexity. We also
show non-constructively that subexponentially samplers exist with
essentially the same parameters as $[0,1]$-valued samplers.

\begin{thm}[Informal version of
	\cref{thm:subgaussian_sampler,cor:subgaussian_samp_nonexplicit}]\label{intro_thm:subgaussian_samplers}
	For every integer $m\in \mathbb N$, $1 > \delta, \eps > 0$, and $\alpha > 0$, there is
	a function $\Samp:\zo^n \to \of{\zo^m}^D$ that is:
	\begin{itemize}
		\item an explicit subgaussian (in
			fact subexponential) sampler with randomness complexity $n = m +
			(1+\alpha)\cdot \log(1/\delta)$ and sample complexity $D =
			\poly(\log(1/\delta),1/\eps)$ (see \cref{thm:subgaussian_sampler})
		\item a non-constructive subexponential sampler with randomness
			complexity $n = m + \log(1/\delta) - \log\log(1/\delta) + O(1)$
			and sample complexity $D = O(\log(1/\delta)/\eps^2)$ (see
			\cref{cor:subgaussian_samp_nonexplicit}).
	\end{itemize}
\end{thm}

\subsection{Randomness extractors}

To prove \cref{intro_thm:subgaussian_samplers}, we develop a
corresponding theory of generalized \emph{randomness extractors} which
we believe is of independent interest. For bounded functions, Zuckerman
\cite{zuckerman_randomnessoptimal_1997} showed that averaging samplers
are essentially equivalent to randomness extractors, and in fact several of
the best-known constructions of such samplers arose as extractor
constructions.  Formally, a randomness extractor is defined as follows:
\begin{defn}[Nisan and Zuckerman \cite{nis_zuc_randomness_1996}]\label{intro_defn:extractor}
	A function $\Ext:\zo^n \times \zo^d\to \zo^m$ is said to be a
	\emph{$(k, \eps)$ extractor} if for every distribution $X$ over $\zo^m$
	satisfying $\max_{x\in \zo^n} \PR{X=x}\leq 2^{-k}$, the distributions
	$\Ext(X, U_d)$ and $U_m$ are $\eps$-close in total variation distance.
	Equivalently, for all $f:\zo^{m} \to [0,1]$ it holds that
	$\E{f(\Ext(X, U_d))} - \E{f(U_m)} \leq \eps$. The number $d$
	is called the \emph{seed length}, and $m$ the \emph{output length}.
\end{defn}

The formulation of \cref{intro_defn:extractor} in terms of
$[0,1]$-valued functions implies that extractors produce an output
distribution that is indistinguishable from uniform by all bounded
functions $f$. It is therefore natural to consider a variant of this
definition for a different set $\F$ of test functions $f:\zo^m \to \R$
which need not be bounded.

\begin{defn}[Special case of \cref{defn:divergence_extractor} using
	\cref{defn:test_functions}]\label{intro_defn:func_extractor}
	A function $\Ext:\zo^n \times \zo^d \to \zo^m$ is said to be a $(k,
	\eps)$ extractor for a set of real-valued functions $\F$ from $\zo^m$
	if for every distribution $X$ over $\zo^m$ satisfying $\max_{x\in
	\zo^n} \PR{X=x}\leq 2^{-k}$ and every $f\in \F$, it holds that
	$\E{f(\Ext(X, U_d))} - \E{f(U_m)} \leq \eps$.
\end{defn}

We show that much of the theory of extractors and samplers carries over
to this more general setting. In particular, we generalize the
connection of Zuckerman \cite{zuckerman_randomnessoptimal_1997} to show
that extractors for a class of functions of $\F$ are also samplers for
that class, along with the converse (though as for total variation
distance, there is some loss of parameters in this direction). Thus, to
construct a subgaussian sampler it suffices (and is preferable) to
construct a corresponding extractor for subgaussian test functions,
which is how we prove \cref{intro_thm:subgaussian_samplers}.

Unfortunately, the distance induced by subgaussian test functions is not
particularly pleasant to work with: for example the point masses on $0$
and $1$ in $\zo$ are $O(1)$ apart, but embedding them in the larger
universe $\zo^m$ leads to distributions which are $\Theta(\sqrt m)$
apart. We solve this problem by constructing extractors for a stronger
notion, the \emph{Kullback--Leibler (KL) divergence}, equivalently,
extractors whose output is required to have very high Shannon entropy.
\begin{defn}[Special case of \cref{defn:divergence_extractor} using
	KL divergence]\label{intro_defn:kl_extractor} A function
	$\Ext:\zo^n \times \zo^d \to \zo^m$ is said to be a $(k, \eps)$
	KL-extractor if for every distribution $X$ over $\zo^m$ satisfying
	$\max_{x\in \zo^n}\allowbreak \PR{X=x}\leq 2^{-k}$ it holds that
	$\KL\diver{\Ext(X, U_d)}{U_m} \leq \eps$, or equivalently
	$\Ent(\Ext(X, U_d)) \geq m - \eps$.
\end{defn}
A strong form of Pinsker's inequality (e.g.\ \cite[Lemma
4.18]{bou_lug_mas_concentration_2013}) implies that a $(k, \eps^2)$
KL-extractor is also a $(k, \eps)$ extractor for subgaussian test
functions. The KL divergence has the advantage that is nonincreasing
under the application of functions (the famous \emph{data-processing
inequality}), and although it does not satisfy a traditional triangle
inequality, it does satisfy a similar inequality when one of the
segments satisfies stronger $\ell_2$ bounds. These properties allow us
to show that the zig-zag product for extractors of Reingold, Wigderson,
and Vadhan \cite{rei_vad_wig_entropy_conference} also works for
KL-extractors, and therefore to construct KL-extractors with seed length
depending on $n$ and $k$ only through the \emph{entropy deficiency} $n -
k$ of $X$ rather than $n$ itself, which in the sampler perspective
corresponds to a sampler with sample complexity depending on the failure
probability $\delta$ rather than the universe size $2^m$. Hence, we
prove \cref{intro_thm:subgaussian_samplers} by constructing
corresponding KL-extractors.

\begin{thm}[Informal version of \cref{thm:high_entropy_kl_extractor}]\label{intro_thm:high_min_entropy_kl}
	For all integers $m$, $1>\delta,\eps > 0$, and $\alpha > 0$ there is an explicit $(k,
	\eps)$ KL-extractor $\Ext:\zo^n \times \zo^d \to \zo^m$ with $n =
	m + (1+\alpha)\cdot \log(1/\delta)$, $k = n - \log(1/\delta)$, and
	$d = O(\log\log(1/\delta) + \log(1/\eps))$.
\end{thm}

Though the above theorem is most interesting in the high min-entropy
regime where $n - k = o(n)$, we also show the existence of
KL-extractors matching most of the existing constructions of total
variation extractors. In particular, we note that extractors for
$\ell_2$ are immediately KL-extractors without loss of parameters,
and also that any extractor can be made a KL-extractor by taking
slightly smaller error, so that the extractors of Guruswami, Umans, and
Vadhan \cite{gur_uma_vad_unbalanced_2009} can be taken to be
KL-extractors with essentially the same parameters.

Furthermore, in addition to our explicit constructions, we also show
non-constructively that KL-extractors (and hence subgaussian
extractors) exist with very good parameters:
\begin{thm}[Informal version of \cref{thm:kl_nonexplicit}]
	For any integers $k < n\in \mathbb N$ and $1>\eps > 0$ there is a
	$(k,\eps)$ KL-extractor $\Ext:\zo^n \times \zo^d \to \zo^m$ with $d
	= \log(n - k) + \log(1/\eps) + O(1)$ and $m = k + d - \log(1/\eps) -
	O(1)$.
\end{thm}

One key thing to note about the nonconstructive KL extractors of the
above theorem is that they incur an entropy loss of only $1\cdot
\log(1/\eps)$, whereas total variation extractors necessarily incur
entropy loss $2\cdot \log(1/\eps)$ by the lower bound of Radhakrishnan
and Ta-Shma \cite{rad_ta-_bounds_2000}.  In particular, by Pinsker's
inequality, $(k, \eps^2)$ KL-extractors with the above parameters are
also optimal $(k, \eps)$ standard (total variation) extractors
\cite{rad_ta-_bounds_2000}, so that one does not lose anything by
constructing a KL-extractor rather than a total variation extractor.
We also remark that the above theorem gives subgaussian samplers with
better parameters than a naive argument that a random function should
directly be a subgaussian sampler, as it avoids the need to take a union
bound over $O(M^M) = O(2^{M\log M})$ test functions (for $M = 2^m$)
which results in additional additive $\log\log$ factors in the
randomness complexity.

In the total variation setting, there are only a couple of methods known
to explicitly achieve optimal entropy loss $2\cdot \log(1/\eps)$, the
easiest of which is to use an extractor which natively has this sort of
loss, of which only three are known: An extractor from random walks over
Ramanujan Graphs due to Goldreich and Wigderson
\cite{gol_wig_tiny_1997}, the Leftover Hash Lemma due to Impagliazzo,
Levin, and Luby \cite{imp_lev_lub_pseudorandom_1989} (see also
\cite{mcinnes_cryptography_1987,ben_bra_rob_privacy_1988}), and the
extractor based on almost-universal hashing of Srinivasan and Zuckerman
\cite{sri_zuc_computing_1999}.  Unfortunately, all of these are $\ell_2$
extractors and so must have seed length linear in $\min(n-k, m)$
(cf.~\cite[Problem 6.4]{vadhan_pseudorandomness_2012}), rather than
logarithmic in $n - k$ as known non-constructively. The other
alternative is to use the generic reduction of Raz, Reingold, and Vadhan
\cite{raz_rei_vad_extracting_2002} which turns any extractor $\Ext$ with
entropy loss $\Delta$ into one with entropy loss $2\cdot
\log(1/\eps)+O(1)$ by paying an additive $O(\Delta + \log (n/\eps))$ in
seed length. We show that all of these $\ell_2$ extractors and the
\cite{raz_rei_vad_extracting_2002} transformation also work to give
KL-extractors with entropy loss $1\cdot \log(1/\eps) + O(1)$, so that
applications which require minimal entropy loss can also use explicit
constructions of KL-extractors.

\subsection{Future directions}

Broadly speaking, we hope that the perspective of KL-extractors will
bring new tools (perhaps from information theory) to the construction of
extractors and samplers. For example, since KL-extractors can have seed
length with dependence on $\eps$ of only $1\cdot \log(1/\eps)$, trying
to explicitly construct a KL-extractor with seed length $1\cdot
\log(1/\eps) + o(\min(n-k,k))$ may also shed light on how to achieve optimal
dependence on $\eps$ in the total variation setting.

In the regime of constant $\eps = \Omega(1)$, we do not have explicit
constructions of subgaussian samplers matching the expander-walk sampler
of Gillman \cite{gillman_chernoff_1998} for $[0,1]$-valued functions,
which achieves randomness complexity $m + O(\log(1/\delta))$ and sample
complexity $O(\log(1/\delta))$, as asked for by B\l{}asiok
\cite{blasiok_optimal_conference}. From the extractor point-of-view, it
would suffice (by the reduction of
\cite{gol_wig_tiny_1997,rei_vad_wig_entropy_conference} that we analyze
for KL-extractors) to construct explicit \emph{linear degree}
KL-extractors with parameters matching the linear degree extractor of
Zuckerman \cite{zuckerman_linear_2007}, i.e.~with seed length $d =
\log(n) + O(1)$ and $m = \Omega(k)$ for $\eps = \Omega(1)$. A
potentially easier problem, since the Zuckerman linear degree extractor
is itself based on the expander-walk sampler, could be to instead match
the parameters of the near-linear degree extractors of Ta-Shma,
Zuckerman, and Safra \cite{ta-_zuc_saf_extractors_2006} based on
Reed--Muller codes, thereby achieving sample complexity
$O(\log(1/\delta)\cdot \poly\log\log(1/\delta))$.

Finally, we hope that KL-extractors can also find uses beyond being
subgaussian samplers and total variation extractors: for example it
seems likely that there are applications (perhaps in coding or
cryptography, cf.~\cite{bar_dod_kra_per_pie_sta_yu_leftover_2011})
where it is more important to have high Shannon entropy in the output
than small total variation distance to uniform, in which case one may be
able to use $(k,\eps)$ KL-extractors with entropy loss only $1\cdot
\log(1/\eps)$ directly, rather than a total variation extractor or
$(k,\eps^2)$ KL-extractor with entropy loss $2\cdot \log(1/\eps)$.

\section{Preliminaries}
\subsection{(Weak) statistical divergences and metrics}

Our results in general will require very few assumptions on notions of
``distance'' between probability distributions, so we will give a
general definition and indicate in our theorems when we need which
assumptions.

\begin{defn}\label{defn:distance}
	A \emph{weak statistical divergence} (or simply \emph{weak
	divergence}) on a finite set $\X$ is a function $\D$ from pairs of
	probability distributions over $X$ to $\R\cup \set{\pm \infty}$.
	We write $\D\diver PQ$ for the value of $\D$ on distributions $P$
	and $Q$. Furthermore
	\begin{enumerate}
		\item If $\D\diver{P}{Q} \geq 0$ with equality iff $P = Q$, then
			$\D$ is \emph{positive-definite}, and we simply call $\D$ a
			\emph{divergence}.
		\item If $\D\diver{P}{Q} = \D\diver{Q}{P}$, then $\D$ is \emph{symmetric}.
		\item If $\D\diver{P}{R} \leq \D\diver{P}{Q} + \D\diver{Q}{R}$, then $\D$
			satisfies the \emph{triangle inequality}.
		\item If $\D\diver{\lambda P_1 + (1-\lambda)P_2}{\lambda Q_1 + (1-\lambda) Q_2}
			\leq \lambda \D\diver{P_1}{Q_1} + (1-\lambda)\D\diver{P_2}{Q_2}$ for all $\lambda
			\in [0,1]$, then $\D$ is \emph{jointly convex}.
			If this holds only when $Q_1 = Q_2$ then $\D$ is \emph{convex
			in its first argument}.
		\item If $\D$ is defined on all finite sets $\mathcal Y$ and for all
			functions $f:\X\to\mathcal Y$ the divergence is nonincreasing
			under $f$, that is $\D\diver{f(P)}{f(Q)} \leq \D\diver{P}{Q}$,
			then $\D$ satisfies the \emph{data-processing inequality}.
	\end{enumerate}
	If $\D$ is positive-definite, symmetric, and satisfies the triangle inequality,
	then it is called a \emph{metric}.
\end{defn}

\begin{exmp}\label{exmp:lp}
	The \emph{$\ell_p$ distance} for $p > 0$ between probability distributions over
	$\X$ is
	\[
		d_{\ell_p}(P, Q) \defeq \of{\sum_{x\in \X} \abs[\big]{P_x - Q_x}^{p}}^{1/p}
	\]
	and is positive-definite and symmetric. Furthermore, for $p\geq 1$ it
	satisfies the triangle inequality (and so is a metric), and is jointly
	convex. The $\ell_p$ distance is nonincreasing in $p$.
\end{exmp}

\begin{exmp}\label{exmp:tv}
	The \emph{total variation distance} is
	\[\tv(P, Q) \defeq \frac12 d_{\ell_1}(P, Q) = \sup_{S\subseteq
		\X}\abs{\PR{P\in S} - \PR{Q \in S}} = \sup_{f\in [0,1]^{\X}}
	\of{\E{f(P)}-\E{f(Q)}}\]
	and is a jointly convex metric that satisfies the data-processing inequality.
\end{exmp}

\begin{exmp}[R\'enyi Divergences \cite{renyi_measures_1961}]\label{defn:renyi_divergence}
	For two probability distributions $P$ and $Q$ over a finite set $\X$,
	the \emph{R\'enyi $\alpha$-divergence} or \emph{R\'enyi divergence of
	order $\alpha$} is defined for real $0 < \alpha \neq 1$ by
	\[
		\D_{\alpha}\diver{P}{Q} \defeq \frac{1}{\alpha -
		1}\log\of{\sum_{x\in \X} \frac{P_x^\alpha}{Q_x^{\alpha - 1}}}
	\]
	where the logarithm is in base $2$ (as are all logarithms in this
	paper unless noted otherwise).  The R\'enyi divergence is continuous
	in $\alpha$ and so is defined by taking limits for
	$\alpha\in\set{0,1,\infty}$, giving for $\alpha = 0$
	the divergence $\D_0\diver PQ
	\defeq \log\of{1/\PR[x\sim Q]{P_x \neq 0}}$, for $\alpha = 1$
	the \emph{Kullback--Leibler (or KL) divergence}
	\[
	\KL\diver{P}{Q} \defeq \D_1\diver{P}{Q} = \sum_{x\in X}P_x\log\frac{P_x}{Q_x},
	\]
	and for $\alpha = \infty$ the \emph{max-divergence}
	$\D_\infty\diver{P}{Q} \defeq \max_{x\in X} \log\frac{P_x}{Q_x}$.
	The R\'enyi divergence is nondecreasing in $\alpha$.
	Furthermore, when $\alpha \leq 1$ the R\'enyi divergence is jointly
	convex, and for all $\alpha$ the R\'enyi divergence satisfies
	the data-processing inequality \cite{van_har_renyi_2014}.

	When $Q = U_{\X}$ is the uniform distribution over the set $\X$, then
	for all $\alpha$, $\D_\alpha\diver{P}{U_{\X}} = \log\abs{\X} -
	\Ent_\alpha(P)$ where $0 \leq \Ent_\alpha(P)\leq \log\abs{\X}$ is
	called the \emph{R\'enyi $\alpha$-entropy of $P$}. For $\alpha = 0$,
	$\Ent_0(P) = \log\abs{\Supp(P)}$ is the \emph{max-entropy of $P$}, for
	$\alpha = 1$, $\Ent_1(P) = \sum_{x\in \X} P_x\log(1/P_x)$ is the
	\emph{Shannon entropy of $P$}, and for $\alpha = \infty$,
	$\Ent_\infty(P) = \min_{x\in \X}\log(1/P_x)$ is the \emph{min-entropy
	of $P$}.

	For $\alpha = 2$, the R\'enyi $2$-entropy can be expressed in terms
	of the $\ell_2$-distance to uniform:
	\[
		\log\abs{\X} - H_2(P) = \D_2\diver P{U_{\X}} = \log\of{1+\abs{\X}
		\cdot d_{\ell_2}(P, U_{\X})^2}
	\]
\end{exmp}

\subsection{Statistical weak divergences from test functions}
\label{sec:metricfromtests}

Zuckerman's connection \cite{zuckerman_randomnessoptimal_1997} between
samplers for bounded functions and extractors for total variation
distance is based on the following standard characterization of total
variation distance as the maximum distinguishing advantage achieved by
bounded functions,
\[\tv(P, Q) = \sup_{f\in [0,1]^{\X}} \E{f(P)} - \E{f(Q)}.\]
By considering an arbitrary class of functions in the supremum, we
get the following weak divergence:

\begin{defn}\label{defn:test_functions}
	Given a finite $\X$ and a set of real-valued functions $\F\subseteq
	\R^{\X}$, the \emph{$\F$-distance on $\X$} between probability
	measures on $\X$ is denoted by $\D^{\F}$ and is defined as
	\[
		\D^{\F}\diver{P}{Q} \defeq \sup_{f\in \F}\of{ \E{f(P)} - \E{f(Q)}}
		= \sup_{f\in \F}\D^{\set{f}}\diver{P}{Q},
	\]
	where we use a superscript to avoid confusion with the
	Csisz\'ar-Morimoto-Ali-Silvey $f$-divergences
	\cite{csiszar_informationstheoretische_1963,morimoto_markov_1963,%
	ali_sil_general_1966}.

	We call the set of functions $\F$ \emph{symmetric} if for all
	$f\in \F$ there is $c\in \R$ and $g\in \F$ such that $g = c - f$,
	and \emph{distinguishing} if for all $P\neq Q$ there
	exists $f\in \F$ with $\D^{\set{f}}\diver PQ > 0$.
\end{defn}
\begin{exmp}
	If $\F = \set{0,1}^{\X}$ or $\F = [0,1]^{\X}$, then $\D^{\F}$ is
	exactly the total variation distance.
\end{exmp}
\begin{rmrk}\label{rmrk:weak_symmetry}
	An equivalent definition of $\F$ being symmetric is that for all $f\in
	\F$ there exists $g\in \F$ with $\D^{\set{g}}\diver{P}{Q} = - D^{\set
	f}\diver PQ = \D^{\set{f}}\diver QP$ for all distributions $P$ and
	$Q$. Hence, one might also consider a weaker notion of symmetry that
	reverses quantifiers, where $\F$ is ``weakly-symmetric'' if for all
	$f\in \F$ and distributions $P$ and $Q$ there exists $g\in \F$ such
	that $\D^{\set{g}}\diver{P}{Q} = -\D^{\set f}\diver PQ=
	\D^{\set{f}}\diver{Q}{P}$. However, such a class ${\F}$ gives exactly
	the same weak divergence $\D^{\F}$ as its ``symmetrization''
	$\overline{\F} = \F\cup\setof{-f}{f\in \F}$, so we do not need to
	introduce this more complex notion.
\end{rmrk}
\begin{rmrk}\label{rmrk:pmf}
	By identifying distributions with their probability mass function, one
	can realize $\E{f(P)} - \E{f(Q)}$ as an inner product $\inner{P - Q,
	f}$.  \Cref{defn:test_functions} can thus be written as $\D^{\F}\diver
	PQ = \sup_{f\in \F}\inner{P - Q, f}$, which is essentially the notion
	of indistinguishability considered in several prior works, (see
	e.g.~the survey of Reingold, Trevisan, Tulsiani, and Vadhan
	\cite{rei_tre_tul_vad_new_2008}), but without requiring all $f$ to be
	bounded.
\end{rmrk}
\begin{rmrk}
	For simplicity, all our probabilistic distributions are given only for
	random variables and distributions over finite sets as this is all we
	need for our application. A more general version of
	\cref{defn:test_functions} has been studied by e.g.~Zolotarev
	\cite{zolotarev_probability_1984} and M\"uller
	\cite{muller_integral_1997} and is commonly used in developments of
	Stein's method in probability.
\end{rmrk}

We now establish some basic properties of $\D^{\F}$.
\begin{lem}\label{lem:test_function_gives_metric}
	Let $\F\subseteq\R^{\X}$ be a set of real-valued functions over
	a finite set $\X$. Then $\D^{\F}$ satisfies the triangle inequality
	and is jointly convex, and
	\begin{enumerate}
		\item if $\F$ is symmetric then $\D^{\F}$ is symmetric and
			\[
				\D^{\F}\diver PQ = \sup_{f\in \F}\abs{\E{f(P)} - \E{f(Q)}} \geq 0,
			\]
		\item if $\F$ is distinguishing	then $\D^{\F}$ is positive-definite,
	\end{enumerate}
	so that if $\F$ is both symmetric and distinguishing then
	$\D^{\F}$ is a jointly convex metric on probability distributions over
	$\X$, in which case we also use the notation $d_{\F}(P, Q) \defeq
	\D^{\F}\diver PQ$.
\end{lem}
\begin{proof}
	The triangle inequality and joint convexity both follow from the
	linearity of each $\D^{\set f}$, as by linearity of expectation, for
	all $f:\X\to \R$ it holds that
	\begin{align*}
		\D^{\set f}\diver PR
		&= \D^{\set f}\diver PQ + \D^{\set f}\diver QR\\
		\D^{\set f}\diver{\lambda P_1 + (1-\lambda)P_2}{\lambda Q_1 + (1-\lambda) Q_2}
		&= \lambda \D^{\set f}\diver{P_1}{Q_1} + (1-\lambda) \D^{\set f}\diver{P_2}{Q_2}.
	\end{align*}
	Upper bounding the terms on the right-hand side by $\D^{\F}$ and
	taking the supremum of the left hand side over $f\in \F$ then gives
	the claims. The symmetry and positive-definite claims are immediate
	from the definitions.
\end{proof}

Furthermore, the notion of dual norm has an appealing interpretation in
this framework via \cref{rmrk:pmf}, generalizing the fact that total
variation distance corresponds to $[0,1]$-valued test functions (or
equivalently that $\ell_1$ distance corresponds to to $[-1,1]$-valued
functions).
\begin{prop}\label{thm:lp_is_bounded_q_moments}
	Let $1\leq p,q \leq \infty$ be H\"older
	conjugates (meaning $1/p + 1/q = 1$), and let
	\[\mathcal{M}_q
		\defeq \set{f:\zo^m \to \R \given \norm{f(U_m)}_q \defeq \E{\abs{f(U_m)}^q}^{1/q} \leq 1}
	\]
	be the set of real-valued functions from $\zo^m$ with bounded
	$q$-th moments. Then $d_{\ell_p} = 2^{-m/q} \cdot d_{\mathcal{M}_q}$,
	in the sense that for all probability distributions $A$ and $B$ over
	$\zo^m$ it holds that $d_{\ell_p}(A, B) = 2^{-m/q}\cdot
	d_{\mathcal{M}_q}(A, B)$.

	In particular, taking $p = 1$ and $q = \infty$ recovers the result
	for $\ell_1$ (equivalently total variation) distance.
\end{prop}
\begin{proof}
	As mentioned this is just the standard fact that the $\ell_p$ and
	$\ell_q$ norms are dual, but for completeness we include a proof
	in our language using the extremal form of H\"older's
	inequality (note that since we are dealing with finite probability
	spaces the extremal equality holds even for $p = \infty$ and $ q= 1$).
	Given probability distributions $A$ and $B$ over $\zo^m$, we have that
	\begin{align*}
		d_{\ell_p}(A, B)
		&= \of{\sum_x \abs{A_x - B_x}^p}^{1/p}\\
		&= 2^{m/p} \E[x\sim U_m]{\abs{A_x - B_x}^p}^{1/p}\\
		&= 2^{m/p} \max_{\substack{f:\zo^m\to\R\\ \norm{f(U_m)}_q\leq 1}} \abs{\E[x\sim U_m]{f(x)(A_x - B_x)}}
		\tag{H\"older's extremal equality}\\
		&= 2^{-m+m/p}\max_{\substack{f:\zo^m\to\R\\ \norm{f(U_m)}_q\leq 1}} \abs{\E{f(A)} - \E{f(B)}}\\
		&= 2^{-m/q} \cdot d_{\mathcal{M}_q}(A, B)\tag{by symmetry of $\mathcal{M}_q$}
	\end{align*}
	as desired.
\end{proof}

\section{Extractors for weak divergences and connections to samplers}
\subsection{Definitions}

We now use this machinery to extend the notion of an extractor due to
Nisan and Zuckerman \cite{nis_zuc_randomness_1996} and the average-case
variant of Dodis, Ostrovsky, Reyzin, and Smith
\cite{dod_ost_rey_smi_fuzzy_2008}.

\begin{defn}[Extends \cref{intro_defn:func_extractor}]\label{defn:divergence_extractor}
	Let $\D$ be a weak divergence on the set $\zo^m$, and $\Ext: \zo^n
	\times \zo^d \to \zo^m$. Then if for all distributions $X$ over
	$\zo^n$ with $\minent(X) \geq k$ it holds that
	\begin{enumerate}
		\item $\D\diver{\Ext(X, U_d)}{U_m} \leq \eps$, then $\Ext$
			is said to be a \emph{$(k, \eps)$ extractor for $\D$}, or
			a \emph{$(k, \eps)$ $\D$-extractor}.
		\item $\E[s\sim U_d]{\D\diver{\Ext(X, s)}{U_m}} \leq \eps$, then $\Ext$
			is said to be a \emph{$(k, \eps)$ strong extractor for $\D$}, or
			a \emph{$(k, \eps)$ strong $\D$-extractor}.
	\end{enumerate}
	Furthermore, if for all joint distributions $(Z, X)$ where $X$ is
	distributed over $\zo^n$ with $\condminent(X|Z) \defeq
	\log\of{1/\E[z\sim Z]{2^{-\minent(X|_{Z=z})}}} \geq k$, it holds that
	\begin{enumerate}\setcounter{enumi}{2}
		\item $\E[z\sim Z]{\D\diver{\Ext(X|_{Z=z}, U_d)}{U_m} \leq \eps}$, then $\Ext$
			is said to be a \emph{$(k, \eps)$ average-case extractor for $\D$}, or
			a \emph{$(k, \eps)$ average-case $\D$-extractor}.
		\item $\E[z\sim Z, s\sim U_d]{\D\diver{\Ext(X|_{Z=z}, s)}{U_m}} \leq \eps$, then $\Ext$
			is said to be a \emph{$(k, \eps)$ average-case strong extractor for $\D$}, or
			a \emph{$(k, \eps)$ average-case strong $\D$-extractor}.
	\end{enumerate}
\end{defn}
\begin{rmrk}
	By taking $\D$ to be the total variation distance we recover the
	standard definitions of extractor and strong extractor due to
	\cite{nis_zuc_randomness_1996} and the definition of average-case
	extractor due to \cite{dod_ost_rey_smi_fuzzy_2008}.

	However, our definitions are phrased slightly differently for strong
	and average-case extractors as an expectation rather than a joint
	distance, that is, for strong average-case extractors we require a
	bound on the expectation $\E[z\sim Z, s\sim
	U_d]{\D\diver{\Ext(X|_{Z=z}, s)}{U_m}}$ rather than a bound on
	$\D\diver{Z, U_d, \Ext(X, U_d)}{Z, U_d, U_m}$.  In our setting, the
	weak divergence $\D$ need not be defined over the larger joint
	universe, but it is defined for all random variables over $\zo^m$. In
	the case of $\tv$ and KL divergence, both definitions are
	equivalent (for KL divergence, this is an instance of the
	\emph{chain rule}).
\end{rmrk}
\begin{rmrk}
	The strong variants of \cref{defn:divergence_extractor} are also
	non-strong extractors assuming the weak divergence $\D$ is convex in
	its first argument, as it is for most weak divergences of interest,
	including the $\ell_p$ norms for $p\geq 1$, all $\D^{\F}$ defined by
	test functions, the KL divergence, R\'enyi divergences for $\alpha
	\leq 1$, and all Csisz\'ar-Morimoto-Ali-Silvey $f$-divergences.
	The average-case variants are always non-average-case extractors by
	taking $Z$ to be independent of $X$.
\end{rmrk}
\begin{rmrk}
	We gave \cref{defn:divergence_extractor} for general weak divergences
	which need not be symmetric, and made the particular choice that the
	output of the extractor was on the left-hand side of the weak
	divergence and that the uniform distribution was on the right-hand
	side. This is motivated by the standard information-theoretic
	divergences such as KL divergence, which require the left-hand
	distribution to have support contained in the support of the
	right-hand distribution, and putting the uniform distribution on the
	right ensures this is always the case. Furthermore, the
	KL divergence to uniform has a natural interpretation as an entropy
	difference, $\KL\diver{P}{U_m} = m - \Ent(P)$ for $\Ent$ the Shannon
	entropy, so that in particular a KL extractor with error $\eps$
	requires the output to have Shannon entropy at least $m - \eps$.  If
	for a weak divergence $\D$ the other direction is more natural, one
	can always reverse the sides by considering the weak divergence
	$\D'\diver QP = \D\diver PQ$.
\end{rmrk}
\begin{rmrk}
	\Cref{defn:divergence_extractor} does not technically need even a weak
	divergence, as it suffices to simply have a measure of distance to
	uniform. However, since weak divergences have minimal constraints, one
	can define a weak divergence from any distance to uniform by ignoring
	the second component (or setting it to be infinite for non-uniform
	distributions).
\end{rmrk}

We also give the natural definition of averaging samplers for arbitrary
classes of functions $\F$ extending \cref{intro_defn:sampler}, along
with the strong variant of Zuckerman
\cite{zuckerman_randomnessoptimal_1997}.

\begin{defn}\label{defn:test_function_sampler}
	Given a class of functions $\F:\zo^m \to \R$, a function $\Samp:\zo^n
	\to \of{\zo^m}^D$ is said to be a \emph{$(\delta,
	\eps)$ strong averaging sampler for $\F$} or a \emph{$(\delta, \eps)$
	strong averaging $\F$-sampler} if for all $f\in \F$, it
	holds that
	\[
		\PR[x\sim U_n]{\E[i\sim U_{[D]}]{f_i\of{\Samp(x)_i}
		- \E{f_i(U_m)}} > \eps} \leq \delta
	\]
	where $[D] = \set{1, \dots, D}$.
	If this holds only when $f_1 = \cdots = f_D$, then it is called a
	\emph{(non-strong) $(\delta, \eps)$ averaging sampler for $\F$} or
	\emph{$(\delta, \eps)$ averaging $\F$-sampler}.  We say that $\Samp$
	is a \emph{$(\delta, \eps)$ strong absolute averaging sampler for $\F$} if it also holds that
	\[
		\PR[x\sim U_n]{\abs[\bigg]{\E[i\sim U_{[D]}]{f_i\of{\Samp(x)_i}
		- \E{f_i(U_m)}}} > \eps} \leq \delta.
	\]
	with the analogous definition for non-strong samplers.
\end{defn}
\begin{rmrk}\label{rmrk:symmetric_test_function_are_abs}
	We separated a single-sided version of the error bound in
	\cref{defn:test_function_sampler} as in
	\cite{vadhan_pseudorandomness_2012}, as it makes the connection
	between extractors and samplers cleaner and allows us to be specific
	about what assumptions are needed. Note that if $\F$ is symmetric then
	every $(\delta, \eps)$ (strong) sampler for $\F$ is a $(2\delta,
	\eps)$ (strong) absolute sampler for $\F$, recovering the standard
	notion up to a factor of $2$ in $\delta$.
\end{rmrk}

\subsection{Equivalence of extractors and samplers}
We now show that Zuckerman's connection
\cite{zuckerman_randomnessoptimal_1997} does indeed generalize to this
broader setting as promised.

\begin{thm}\label{thm:func_ext_are_samplers}
	Let $\Ext:\zo^n\times\zo^d \to \zo^m$ be an
	$(n-\log(1/\delta),\eps)$-extractor (respectively strong extractor)
	for the weak divergence $\D^{\F}$ defined by a class of test functions
	$\F:\zo^m\to\R$ as in \cref{defn:test_functions}. Then the function
	$\Samp:\zo^n \to \of{\zo^m}^{D}$ for $D = 2^d$ defined by $\Samp(x)_i
	= \Ext(x, i)$ is a $(\delta, \eps)$-sampler (respectively strong
	sampler) for $\F$.
\end{thm}
\begin{proof}
	The proof is essentially the same as that of
	\cite{zuckerman_randomnessoptimal_1997}.

	\newcommand{\Bfg}{B_{f_1,\dots,f_D}}
	\newcommand{\Xg}{X}
	Fix a collection of test functions $f_1,\dots,f_D\in \F$, where if
	$\Ext$ is not strong we restrict to $f_1 = \cdots = f_D$, and let
	$\Bfg\subseteq \zo^n$ be defined as
	\begin{align*}
		\Bfg
		&\defeq \setof{x\in \zo^n}{\E[i\sim U_{[D]}]{f_i\of{\Ext(x,i)} -
		\E{f_i(U_m)}} > \eps}\\
		&= \setof{x\in \zo^n}{\E[i\sim U_{[D]}]{\D^{\set{f_i}}\diver{U_{\set{\Ext(x,i)}}}{U_m}}
		> \eps},
	\end{align*}
	where $U_{\set z}$ is the point mass on $z$. Then if $\Xg$ is uniform over $\Bfg$, we have
	\begin{align*}
		\eps
		&<
		\E[x\sim \Xg]{\E[i \sim U_{[D]}]{f_i\of{\Ext(x,i)} - \E{f_i(U_m)}}}\\
		&= \E[i \sim U_{[D]}]{\D^{\set{f_i}}\diver{\Ext(\Xg,i)}{U_m}}\\
		&=
		\begin{cases}
			\D^{\set{f_1}}\diver{\Ext(\Xg, U_d)}{U_m}&\text{if }f_1=\cdots=f_D\\
			\E[i\sim U_{[D]}]{\D^{\set{f_i}}\diver{\Ext(\Xg, i)}{U_m}}&\text{always}
		\end{cases}\\
		&\leq
		\begin{cases}
			\D^{\F}\diver{\Ext(\Xg, U_d)}{U_m}&\text{if }f_1=\cdots=f_D\\
			\E[i\sim U_{[D]}]{\D^{\F}\diver{\Ext(\Xg, i)}{U_m}}&\text{always}
		\end{cases}
	\end{align*}
	Since $\Ext$ is an $(n-\log(1/\delta),\eps)$-extractor (respectively
	strong extractor) for $\D^{\F}$ we must have $\minent(\Xg) < n -
	\log(1/\delta)$. But $\minent(\Xg) = \log\abs{\Bfg}$ by definition, so
	we have $\abs{\Bfg} < \delta 2^n$. Hence, the probability that a
	random $x\in \zo^n$ lands in $\Bfg$ is less than $\delta$, and since
	$\Bfg$ is exactly the set of seeds which are bad for $\Samp$, this
	concludes the proof.
	\let\Bfg\undefined
	\let\Xg\undefined
\end{proof}
\begin{rmrk}
	H\"older's inequality implies that an extractor for $\ell_p$ with
	error $\eps\cdot 2^{-m(p-1)/p}$ is also an $\ell_1$ extractor and thus
	$[-1,1]$-averaging sampler with error $\eps$.
	\Cref{thm:lp_is_bounded_q_moments,thm:func_ext_are_samplers} show that
	they are in fact samplers for the much larger class of functions
	$\mathcal{M}_{p/(p-1)}$ with bounded $p/(p-1)$ moments (rather than
	just $\infty$ moments), also with error $\eps$.
\end{rmrk}

Furthermore, if all the functions in $\F$ have bounded deviation from
their mean (for example, subgaussian functions from $f:\zo^m \to \R$
have such a bound of $O(\sqrt{m})$ by the tail bounds from
\cref{lem:subgaussian_properties}), then we also have a partial converse
that recovers the standard converse in the case of total variation
distance.

\begin{thm}\label{thm:samplers_are_extractors}
	Let $\F$ be a class of functions $\F\subset \zo^m \to \R$ with finite
	\emph{maximum deviation from the mean}, meaning $\maxdev(\F) \defeq
	\sup_{f\in \F}\max_{x\in \zo^n}\of[\big]{f(x) - \E{f(U_m)}} < \infty$.
	Then given a $(\delta, \eps)$ $\F$-sampler (respectively $(\delta,\eps)$
	strong $\F$-sampler) $\Samp:\zo^n\to \of{\zo^m}^D$, the function
	$\Ext:\zo^n \times \zo^d \to \zo^m$ for $d = \log D$ defined by $\Ext(x,
	i) = \Samp(x)_i$ is a $\of[\big]{k, \eps+\delta\cdot 2^{n-k}\cdot
	\maxdev(\F)}$ $\D^{\F}$-extractor (respectively strong
	$\D^{\F}$-extractor) for every $0 \leq k \leq n$.

	In particular, $\Ext$ is an $\of[\big]{n - \log(1/\delta) +
	\log(1/\eta), \eps+\eta\cdot \maxdev(\F)}$ average-case
	$\D^{\F}$-extractor (respectively strong average-case
	$\D^{\F}$-extractor) for every $\delta\leq \eta \leq 1$.
\end{thm}
\begin{proof}
	Again the proof is analogous to the one in \cite{zuckerman_randomnessoptimal_1997}.

	Fix a distribution $X$ over $\zo^m$ with $\minent(X)\geq k$ and a
	collection of test functions $f_1,\dots,f_D\in \F$, where if $\Samp$
	is not strong we restrict to $f_1 = \cdots = f_D$.
	\newcommand{\Bf}{B_{f_1,\dots, f_D}}
	Then since $\Samp$ is a $(\delta, \eps)$ $\F$-sampler, we know that the set
	of seeds for which the sampler is bad must be small. Formally, the set
	\begin{align*}
		\Bf
		&\defeq \set{x\in \zo^n \given \E[i\sim U_d]{f_i(\Samp(x)_i) - \E{f_i(U_m)}} > \eps}\\
		&= \set{x\in \zo^n \given \E[i\sim U_d]{f_i\of{\Ext(x,i)}- \E{f_i(U_m)}} > \eps}
	\end{align*}
	has size $\abs{\Bf} \leq \delta 2^n$.  Thus, since $X$ has min-entropy
	at least $k$ we know $\PR{X\in \Bf} \leq 2^{-k}\cdot \delta 2^n$, so we have
	\begin{align*}
		\E[i\sim U_d]{\E{f_i(\Ext(X, i)) - \E{f_i(U_m)}}}&\\
		&\hspace{-9em}= \E[X]{\E[i\sim U_d]{f_i(\Ext(X, i))-\E{f_i(U_m)}}}\\
		&\hspace{-9em}=\PR{X\in \Bf}\cdot \CE[X]{\E[i\sim U_d]{f_i(\Ext(X, i))-\E{f_i(U_m)}}}{X\in \Bf}\\
		&\hspace{-7em}+ \PR{X\not\in \Bf}\cdot \CE[X]{\E[i\sim U_d]{f_i(\Ext(X, i))-\E{f_i(U_m)}}}{X\not\in \Bf}\\
		&\hspace{-9em}\leq \PR{X\in \Bf}\cdot \maxdev(\F) + \PR{X\not\in\Bf}\cdot \eps\\
		&\hspace{-9em}\leq  2^{-k}\cdot \delta2^n \cdot \maxdev(\F) + \eps
	\end{align*}
	completing the proof of the main claim. The ``in particular'' statement follows
	since if $(Z, X)$ are jointly distributed with $\condminent(X|Z) \geq n - \log(1/\delta)
	+ \log(1/\eta)$ we have
	\[
		\E[z\sim Z]{\eps + \delta\cdot 2^{n-\minent(X|_{Z=z})}\cdot \maxdev(\F)}
			= \eps + \delta\cdot2^{n-\condminent(X|Z)}\cdot \maxdev(\F) \leq \eps + \eta\cdot \maxdev(\F)
	\]
	by definition of conditional min-entropy.
	\let\Bf\undefined
\end{proof}

\subsection{All extractors are average-case}

Under a similar boundedness condition for general weak divergences, we
can recover the standard fact that all extractors are average-case
extractors under a slight loss of parameters (the same loss as achieved
by Dodis, Ostrovsky, Reyzin, and Smith \cite{dod_ost_rey_smi_fuzzy_2008}
for the case of total variation distance). More interestingly, if the
weak divergence is given by $\D^{\F}$ for a symmetric class of (possibly
unbounded) functions $\F$, we can also generalize and recover the result
of Vadhan \cite[Problem 6.8]{vadhan_pseudorandomness_2012} that shows
that a $(k,\eps)$ extractor (for total variation) is a $(k, 3\eps)$
average-case extractor without any other loss.

\begin{thm}\label{thm:ext_are_average_case}
	Let $\D$ be a bounded weak divergence over $\zo^m$, meaning that
	\[0\leq \norm{\D}_{\infty}
	\defeq \sup_{P\text{ on }\zo^m} \D\diver{P}{U_m} < \infty.\]
	Then a $(k, \eps)$-extractor for $\D$ (respectively strong extractor)
	$\Ext:\zo^n\times \zo^d \to \zo^m$ is also a $(k + \log(1/\eta), \eps
	+ \eta \cdot \norm{\D}_\infty)$ average-case-extractor for $\D$
	(respectively strong average-case-extractor) for any $0 < \eta \leq
	1$.
\end{thm}
\begin{proof}
	The proof is analogous to that of \cite{dod_ost_rey_smi_fuzzy_2008}.
	We prove it only for non-strong extractors, the proof for strong
	extractors is completely analogous by adding more expectations.

	For jointly distributed random variables $(Z, X)$ such that
	$\condminent(X|Z) \geq k + \log(1/\eta)$, we have by \cite[Lemma
	2.2]{dod_ost_rey_smi_fuzzy_2008} that the probability that $\PR[z\sim
	Z]{\minent(X|_{Z=z}) < k} \leq \eta$. Thus
	\begin{align*}
		&\E[z\sim Z]{\D\diver{\Ext\of[\big]{X|_{Z=z}, U_d}}{U_m}}\\
		&\qquad=
		\PR[z\sim Z]{\minent(X|_{Z=z}) < k}\cdot
		\CE[z\sim Z]{\D\diver{\Ext\of[\big]{X|_{Z=z},
		U_d}}{U_m}}{\minent(X|_{Z=z}) < k}\\
	 &\qquad\qquad+
		\PR[z\sim Z]{\minent(X|_{Z=z}) \geq k}\cdot
		\CE[z\sim Z]{\D\diver{\Ext\of[\big]{X|_{Z=z},
		U_d}}{U_m}}{\minent(X|_{Z=z}) \geq k}\\
		&\qquad\leq \eta \cdot \norm{D}_{\infty} + 1\cdot \eps\tag*{\qedhere}
	\end{align*}
\end{proof}

\begin{thm}\label{thm:func_ext_are_average}
	Let $\F$ be a symmetric class of test functions and $\Ext:\zo^n \times
	\zo^d \to \zo^m$ be a $(k, \eps)$ extractor (respectively strong
	extractor) for $\D^{\F}$, where $k$ is at most $n - 1$. Then
	$\Ext$ is an $(k, 3\eps)$ average-case extractor
	(respectively strong average-case extractor) for $\D^{\F}$.
\end{thm}
\begin{rmrk}
	\Cref{thm:func_ext_are_average} also applies to extractors for the
	$\ell_p$ norms via \cref{thm:lp_is_bounded_q_moments}.
\end{rmrk}

The proof of \cref{thm:func_ext_are_average} follows the strategy
outlined by Vadhan \cite[Problem 6.8]{vadhan_pseudorandomness_2012}. We
first isolate the following key lemma which shows that any extractor
with error that gracefully decays with lower min-entropy is average-case
with minimal loss of parameters, as opposed to
\cref{thm:ext_are_average_case} which used a worst-case error bound when
the min-entropy is low.
\begin{lem}\label{lem:graceful_ext_is_average}
	Let $\Ext:\zo^n \times \zo^d \to \zo^m$ be a $(k, \eps)$ extractor
	(respectively strong extractor) for $\D$ such that for every $0\leq
	t\leq k$, $\Ext$ is also a $(k - t, 2^{t+1}\cdot \eps)$ extractor
	(respectively strong extractor) for $\D$. Then $\Ext$ is a $(k,
	3\eps)$ average-case extractor (respectively strong average-case
	extractor) for $\D$.
\end{lem}
\begin{proof}
	We prove this for strong extractors, the non-strong case is analogous.
	For every $(Z, X)$ with $X$ distributed on $\zo^n$ and
	$\condminent(X|Z) \geq k$, we have
	\begin{align*}
		\E[z\sim Z,s\sim U_d]{\D\diver{\Ext(X|_{Z=z}, s)}{U_m}}
		&=
		\E[z\sim Z]{\E[s\sim U_d]{\D\diver{\Ext(X|_{Z=z}, s)}{U_m}}}\\
		&\leq \E[z\sim Z]{\begin{cases}\eps&\text{if }\minent(X|_{Z=z})\geq k\\ 2^{k - \minent(X|_{Z=z})+1}\cdot \eps &\text{otherwise}\end{cases}}\\
		&\leq \eps \cdot \E[z\sim Z]{1 + 2^{k - \minent(X|_{Z=z})+1}}
		\leq 3\eps
	\end{align*}
	where the last inequality follows from the fact that $\E[z\sim
	Z]{2^{-\minent(X|_{Z=z})}} = 2^{-\condminent(X|Z)}$ by definition of
	conditional min-entropy.
\end{proof}

\begin{proof}[Proof of \cref{thm:func_ext_are_average}]
	By the previous lemma, it suffices to prove that for every $t\geq 0$,
	$\Ext$ is a $(k - t, (2^{t+1}-1)\cdot \eps)$ extractor (respectively
	strong extractor) for $\D^{\F}$.  Since $\D^{\F}$ is convex in its
	first argument by \cref{lem:test_function_gives_metric}, following
	Chor and Goldreich \cite{cho_gol_unbiased_1988} it is enough
	to consider only distributions with min-entropy $k - t$ that are
	supported on a set of at most $2^{n-1}$. Fix such a distribution $X$
	and a collection of test functions $f_1, \dotsc, f_D \in \F$ with $f_1
	= \cdots = f_D$ if $\Ext$ is not strong.  Then since $X$ is supported
	on a set of size at most $2^{n-1}$, the distribution $Y$ that is
	uniform over the complement of $\Supp(X)$ has min-entropy at least $n
	- 1 \geq k$, and furthermore the mixture $2^{-t}X + (1-2^{-t})Y$ has
	min-entropy at least $k$.  Hence, as $\Ext$ is a $(k,\eps)$ extractor
	(respectively strong extractor) for $\D^{\F}$,
	\begin{align*}
		\eps
			&\geq \E[i\sim U_{[D]}]{\D^{\set{f_i}}\diver{\Ext\of{2^{-t}X + (1-2^{-t})Y, i}}{U_m}}\\
				&= 2^{-t}\E[i\sim U_{[D]}]{\D^{\set{f_i}}\diver{\Ext\of{X,i}}{U_m}}
			+ (1-2^{-t})\E[i\sim U_{[D]}]{\D^{\set{f_i}}\diver{\Ext\of{Y,i}}{U_m}}\\
				&= 2^{-t}\E[i\sim U_{[D]}]{\D^{\set{f_i}}\diver{\Ext\of{X,i}}{U_m}}
			- (1-2^{-t})\E[i\sim U_{[D]}]{\D^{\set{c_i-f_i}}\diver{\Ext\of{Y,i}}{U_m}}\\
				&\geq 2^{-t}\E[i\sim U_{[D]}]{\D^{\set{f_i}}\diver{\Ext\of{X,i}}{U_m}}
				- (1 - 2^{-t})\cdot  \eps\tag{since $\minent(Y)\geq k$}\\
		\of{2^{t+1} - 1}\cdot \eps &\geq\E[i\sim U_{[D]}]{\D^{\set{f_i}}\diver{\Ext\of{X,i}}{U_m}}
	\end{align*}
	where $c_i\in \R$ is such that $c_i - f_i \in \F$ as guaranteed to
	exist by the symmetry of $\F$.
\end{proof}

\section{Subgaussian distance and connections to other notions}

Now that we've introduced the general machinery we need, we can go back
to our motivation of subgaussian samplers. We will need some standard
facts about subgaussian and subexponential random variables, we
recommend the book of Vershynin \cite{vershynin_highdimensional_2018}
for an introduction.

\begin{defn}\label{defn:subgaussian}
	A real-valued mean-zero random variable $Z$ is said to be \emph{subgaussian
	with parameter $\sigma$} if for every $t \in \R$ the moment generating
	function of $Z$ is bounded as
	\[
		\ln\E{e^{t Z}} \leq \frac{t^2\sigma^2}{2}.
	\]
	If this is only holds for $\abs{t} \leq b$ then $Z$ is said to be
	\emph{$(\sigma, b)$-subgamma}, and if $Z$ is $(\sigma,
	1/\sigma)$-subgamma then $Z$ is said to be \emph{subexponential with
	parameter $\sigma$}.
\end{defn}
\begin{rmrk}
	There are many definitions of subgaussian (and especially
	subexponential) random variables in the literature, but they are all
	equivalent up to constant factors in $\sigma$ and only affect
	constants already hidden in big-$O$'s.
\end{rmrk}
\begin{lem}\label{lem:subgaussian_properties}
	Let $Z$ be a real-valued random variable. Then
	\begin{enumerate}
		\item (Hoeffding's lemma) If $Z$ is bounded in the interval
			$[0, 1]$, then $Z - \E{Z}$ is subgaussian with parameter
			$1/2$.
		\item If $Z$ is mean-zero, then $Z$ is subgaussian (respectively
			subexponential) with parameter $\sigma$ if and only if $cZ$ is
			subgaussian (respectively subexponential) with parameter
			$\abs{c}\sigma$ for every $c\neq 0$.
	\end{enumerate}
	Furthermore, if $Z$ is mean-zero and subgaussian with parameter $\sigma$,
	then
	\begin{enumerate}
		\item For all $t > 0$, $\max\of[\big]{\PR{Z > t}, \PR{Z < -t}}
			\leq e^{-t^2/2\sigma^2}$.
		\item $\norm{Z}_p \defeq \E{\abs{Z}^p}^{1/p} \leq
			2\sigma \sqrt{p}$ for all $p \geq 1$.
		\item $Z$ is subexponential with parameter $\sigma$.
	\end{enumerate}
\end{lem}

We are now in a position to formally define the \emph{subgaussian
distance}.
\begin{defn}\label{defn:subgaussian_distance}
	For every finite set $\X$, we define the set $\cG_{\X}$ of
	\emph{subgaussian test functions on $\X$} (respectively the set
	$\cE_{\X}$ of \emph{subexponential test functions on $\X$}) to be the
	set of functions $f:\X\to \R$ such that the random variable
	$f(U_{\X})$ is mean-zero and subgaussian (respectively subexponential)
	with parameter $1/2$.
	Then $\cG_{\X}$ and $\cE_{\X}$ are symmetric and distinguishing, so by
	Lemma~\ref{lem:test_function_gives_metric} the respective distances
	induced by $\cG_{\X}$ and $\cE_{\X}$ are jointly convex metrics called the
	\emph{subgaussian distance} and \emph{subexponential distance}
	respectively and are denoted as $d_{\cG}(P, Q)$ and $d_{\cE}(P, Q)$.
\end{defn}
\begin{rmrk}\label{rmrk:whyvarhalf}
	We choose subgaussian parameter $1/2$ in
	\cref{defn:subgaussian_distance} as by Hoeffding's lemma, all
	functions $f:\zo^m \to [0,1]$ have that $f(U_m) - \E{f(U_m)}$ is
	subgaussian with parameter $1/2$, so this choice preserves the same
	``scale'' as total variation distance. However, the choice of
	parameter is essentially irrelevant by linearity, as different choices
	of parameter simply scale the metric $d_{\cG}$.
\end{rmrk}

Note that absolute averaging samplers for $\cG_{\zo^m}$ from
\cref{defn:test_function_sampler} are exactly subgaussian samplers as
defined in the introduction.  Thus, by
\cref{rmrk:symmetric_test_function_are_abs,thm:func_ext_are_samplers}, to
construct subgaussian samplers it is enough to construct extractors for
the subgaussian distance $d_{\cG}$.

\subsection{Composition}\label{sec:composition}
Unfortunately, the subgaussian distance has a major disadvantage
compared to total variation distance that complicates extractor
construction: it does not satisfy the data-processing inequality, that is, there are
probability distributions $P$ and $Q$ over a set $A$ and a function $f:
A\to B$ such that
\[ d_{\cG}(f(P), f(Q)) \not\leq d_{\cG}(P, Q). \]
This happens because subgaussian distance is defined by functions which
are required to be subgaussian only with respect to the \emph{uniform
distribution}. A simple explicit counterexample comes from taking
$f:\zo^{1} \to \zo^{m}$ defined by $x \mapsto (x, 0^{m-1})$ and taking
$P$ to be the point mass on $0$ and $Q$ the point mass on $1$.  Their
subgaussian distance in $\zo^1$ is obviously $O(1)$, but the subgaussian
distance of $f(P)$ and $f(Q)$ in $\zo^m$ is $\Theta(\sqrt{m})$.

The reason this matters because a standard operation (cf.~Nisan and
Zuckerman \cite{nis_zuc_randomness_1996}; Goldreich and Wigderson
\cite{gol_wig_tiny_1997}; Reingold, Vadhan, and Wigderson
\cite{rei_vad_wig_entropy_conference}) in the construction of samplers
and extractors for bounded functions is to do the following: given
extractors
\begin{align*}
	\Ext_{out}&:\zo^n \times \zo^d \to \zo^m\\
	\Ext_{in}&:\zo^{n'} \times \zo^{d'} \to \zo^d,
\end{align*}
define $\Ext:\zo^{n+n'} \times \zo^{d'} \to \zo^m$ by
\[
	\Ext\of[\big]{(x,y),s} = \Ext_{out}\of[\big]{x, \Ext_{in}\of{y,s}}.
\]
The reason this works for total variation distance is exactly the
data-processing inequality: if $Y$ has enough min-entropy given $X$, then
$\Ext_{in}(Y, U_{d'})$ will be close in total variation distance to
$U_d$, and by the data-processing inequality for total variation
distance this closeness is not lost under the application of
$\Ext_{out}$. The assumption that $Y$ has min-entropy given $X$ means
that $(X, Y)$ is a so-called \emph{block-source}, and is implied by $(X,
Y)$ having enough min-entropy as a joint distribution. From the sampler
perspective, this construction uses the inner sampler $\Ext_{in}$ to
subsample the outer sampler. On the other hand, for subgaussian
distance, the distribution $\Ext_{in}\of{Y,U_{d'}}$ can be $\eps$-close
to uniform but still have some element with excess probability mass
$\Omega(\eps/\sqrt{d})$, and this element (seed) when mapped by
$\Ext_{out}$ can retain\footnote{Given a subgaussian extractor $\Ext$
	with $d\geq \log(m/\eps)$, adding a single extra seed $*$ to $\Ext$
	such that $\Ext(x, *) = 0^m$ results in a subgaussian extractor with
error at most $2^{-d}\cdot \sqrt{2m} + \eps \leq 3\eps$ by convexity of
$d_{\cG}$ and the fact that $\norm{d_{\cG_{\zo^m}}}_{\infty} <
\sqrt{2m}$.} this excess mass in $\zo^m$, which results in subgaussian
distance $\Theta(\eps\sqrt{m/d}) \gg \eps$. Similarly, from the sampler
perspective, even when the outer sampler $\Ext_{out}$ is a good
subgaussian sampler for $\zo^m$, there is no reason that a good
subgaussian sampler $\Ext_{in}$ for $\zo^d$ the seeds of $\Ext_{out}$
will preserve the larger sampler property when $m\gg d$.

Thus, since this composition operation is needed to construct high-min
entropy extractors with the desired seed length even for total variation
distance, to construct such extractors for subgaussian distance we need
to bypass this barrier. The natural approach is to construct extractors
for a better-behaved weak divergence that bounds the subgaussian
distance.

\begin{rmrk}
	Similar reasoning shows that if $\Ext$ is a strong $(k, \eps)$
	subgaussian extractor, then it is not necessarily the case that the
	function $(x, s) \mapsto (s, \Ext(x, s))$ that prepends the seed to
	the output is a (non-strong) $(k, \eps)$ subgaussian extractor (in
	contrast to extractors for total variation distance), though the
	converse does hold.
\end{rmrk}

\subsection{Connections to other weak divergences}

Therefore, to aid in extractor construction, we show how $d_{\cG}$
relates to other statistical weak divergences.

Most basically, the subgaussian distance over $\zo^m$ differs from
total variation distance up to a factor of $O(\sqrt{m})$.
\begin{lem}\label{lem:tv_vs_dg}
	Let $P$ and $Q$ be distributions on $\zo^m$. Then
	\[ \tv(P, Q) \leq d_{\cG}(P, Q) \leq \sqrt{2\ln 2 \cdot m} \cdot \tv(P, Q) \]
\end{lem}
\begin{proof}
	That $\tv \leq d_{\cG}$ is immediate from Hoeffding's lemma and the
	discussion in \cref{rmrk:whyvarhalf}. The reverse bound holds since
	any subgaussian function takes values at most $\sqrt{\ln2/2 \cdot m}$
	away from the mean by the tail bounds
	from part 3 of \cref{lem:subgaussian_properties}, and so any subgaussian test
	function $f$ has the property that $1/2 + f/\sqrt{2 \ln 2\cdot m}$
	is $[0,1]$-valued and thus lower bounds the total variation distance.
\end{proof}

While this allows constructing subgaussian extractors and samplers from
total variation extractors, as discussed in the introduction the fact
that the upper bound depends on $m$ leads to suboptimal bounds. By starting
with a stronger measure of error, we pay a much smaller penalty.
\begin{lem}\label{lem:lp_vs_dg}
	Let $P$ and $Q$ be distributions on $\zo^m$. Then for every
	$\alpha > 0$
	\begin{align*}
		2\tv(P, Q) = d_{\ell_1}(P, Q) &\leq 2^{m\alpha/(1+\alpha)}
		\cdot d_{\ell_{1+\alpha}}(P, Q)\\
		d_{\cG}(P, Q) &\leq 2^{m\alpha/(1+\alpha)} \sqrt{1+\frac1\alpha}
		\cdot d_{\ell_{1+\alpha}}(P, Q)
	\end{align*}
\end{lem}
In particular, that there is only an additional $\sqrt{1+1/\alpha}$
factor when moving to subgaussian distance compared to total
variation, which in particular does not depend on $m$ and is constant
for constant $\alpha$.
\begin{proof}
	By \cref{thm:lp_is_bounded_q_moments}, for any function $f:\zo^m \to
	\R$ it holds that
	\[\D^{\set f}\diver PQ \leq \norm{f(U_m)}_{1+\frac1\alpha} \cdot
	d_{\mathcal{M}_{1+\frac1\alpha}}(P, Q) =
	\norm{f(U_m)}_{1+\frac1\alpha}\cdot 2^{m\alpha/(1+\alpha)}
	\cdot
	d_{\ell_{1+\alpha}}(P, Q).\]
	The result follows since $[-1,1]$-valued functions $f$ satisfy moment
	bounds $\norm{f(U_m)}_q \leq 1$ for all $q\geq 1$, and functions $f$
	which are subgaussian satisfy moment bounds $\norm{f(U_m)}_{q} \leq
	\sqrt{q}$ by \cref{lem:subgaussian_properties}.
\end{proof}

One downside of starting with bounds on $\ell_{1+\alpha}$ is that,
extending a well-known linear seed length linear bound for
$\ell_2$-extractors (e.g.~\cite[Problem
6.4]{vadhan_pseudorandomness_2012}), we show in \cref{cor:l1pa-lb} that
for every $1 > \alpha > 0$, there is a constant $c_\alpha > 0$ such any
$\ell_{1+\alpha}$ extractor with error smaller than $c_\alpha\cdot
2^{-m\alpha/(1 + \alpha)}$ requires seed length linear in $\alpha \cdot
\min(n-k, m)$, for $n - k$ the entropy deficiency and $m$ the output
length.  One might hope that sending $\alpha$ to $0$ would eliminate
this linear lower bound but still bound the subgaussian distance, but
phrased this way sending $\alpha$ to $0$ just results in a total
variation extractor.

However, with a shift in perspective essentially the same approach
works: by \cref{defn:renyi_divergence}, $d_{\ell_2}(P, U_m) \leq \eps
\cdot 2^{-m/2}$ implies $\D_2\diver P{U_m} \leq \eps^2/\ln 2$, and there
is an analogous linear seed length lower bound on constant error
$\D_{1+\alpha}$ extractors for every $\alpha > 0$.  In this case,
however, sending $\alpha$ to $0$ results in the \emph{KL divergence},
which does upper bound the subgaussian distance, and in fact with the
same parameters as for total variation distance.
\begin{lem}\label{lem:dg_vs_kl}
	Let $P$ be a distribution on $\zo^m$. Then
	\begin{align*}
		d_{\cG}(P, U_m)
		&\leq \sqrt{\frac{\ln2}{2}\cdot \KL\diver{P}{U_m}}\\
		d_{\cE}(P, U_m)
		&\leq
		\begin{cases}
			\sqrt{\frac{\ln2}{2}\cdot \KL\diver{P}{U_m}} &\text{if }\KL\diver{P}{U_m} \leq \frac{1}{2\ln 2}\\
			\frac{\ln2}{2}\cdot \KL\diver{P}{U_m} + \frac{1}{4} &\text{if }\KL\diver{P}{U_m} > \frac{1}{2\ln 2}
		\end{cases}\\
	\intertext{where these bounds are concave in $\KL\diver P{U_m}$.
	In the reverse direction, it holds that}
		\KL\diver P{U_m} &\leq m\cdot \tv(P, U_m) + h(\tv(P, U_m))
	\end{align*}
	where $h(x) = x\log(1/x) + (1-x)\log(1/(1-x))$ is the (concave) binary
	entropy function.
\end{lem}
\begin{proof}
	The upper bound on subgaussian distance follows from a general form of
	Pinsker's inequality as in \cite[Lemma
	4.18]{bou_lug_mas_concentration_2013}, but for the extension to
	subexponential functions we reproduce its proof here, based on the
	Donsker--Varadhan ``variational'' formulation of KL divergence
	\cite{don_var_asymptotic_1976} (cf.~\cite[Corollary
	4.15]{bou_lug_mas_concentration_2013})
	\[
		\KL\diver{P}{U_m} = \frac{1}{\ln 2} \cdot \sup_{g:\zo^m\to \R}\of{\E{g(P)} - \ln\E{e^{g(U_m)}}}.
	\]
	Now if $f:\zo^m \to \R$ satisfies $\E{f(U_m)} = 0$, then by letting
	$g(x) = t\cdot f(x)$, this implies
	\[
		\E{f(P)} - \E{f(U_m)}
		=
		\frac{1}{t} \cdot \E{g(P)}
		\leq
		\frac{\ln 2\cdot \KL\diver{P}{U_m} + \ln\E{e^{t\cdot f(U_m)}}}{t}
	\]
	for all $t > 0$. Thus, when $\ln \E{e^{t\cdot f(U_m)}} \leq t^2/8$, we
	have $\E{f(P)} - \E{f(U_m)} \leq \ln2\cdot\KL\diver{P}{U_m}/t + t/8$.

	Then since subgaussian random variables satisfy such a bound for all
	$t$, we can make the optimal choice $t = \sqrt{8\ln2\cdot
	\KL\diver{P}{U_m}}$ to get the claimed bound on $d_{\cG}$.  For
	subexponential random variables, which satisfy such a bound only for
	$\abs t \leq 2$, we choose $t = \min(\sqrt{8\ln 2\cdot
	\KL\diver{P}{U_m}}, 2)$, which gives
	\[
		d_{\cE}(P, U_m)
		\leq
		\begin{cases}
			\sqrt{\frac{\ln2}{2}\cdot \KL\diver{P}{U_m}} &\text{if }\KL\diver{P}{U_m} \leq \frac{1}{2\ln 2}\\
			\frac{\ln2}{2}\cdot \KL\diver{P}{U_m} + \frac{1}{4} &\text{if }\KL\diver{P}{U_m} > \frac{1}{2\ln 2}
		\end{cases}
	\]
	as desired. The concavity of this bound follows by noting that it has
	a continuous and nonincreasing derivative.

	For the reverse inequality, we use a bound on the difference in
	entropy between distributions $P$ and $Q$ on a set of size $S$ which
	states
	\[
		\abs{H(P) - H(Q)} \leq \lg\of{S-1} \cdot \tv(P, Q) + h(\tv(P, Q)).
	\]
	This inequality is a simple consequence of Fano's inequality as noted
	by Goldreich and Vadhan \cite[Fact B.1]{gol_vad_comparing_1999}, and
	implies the desired result by taking $Q = U_m$ as $\KL\diver{P}{U_m} =
	H(U_m) - H(P)$ and $\abs{\zo^m} = 2^m$.
\end{proof}
\begin{rmrk}
	There are sharper upper bounds on the KL divergence than given
	in \cref{lem:dg_vs_kl}, such as the bound of Audenaert and Eisert
	\cite[Theorem 6]{aud_eis_continuity_2005}, but the bound we use has
	the advantage of being defined for the entire range of the
	total variation distance and being everywhere concave.
\end{rmrk}

\section{Extractors for KL divergence}

By \cref{lem:dg_vs_kl}, the subgaussian distance can be bounded in terms
of the KL divergence to uniform, so by the following easy lemma to
construct subgaussian extractors it suffices to construct extractors for
KL divergence.
\begin{lem}\label{lem:divergence_bounds}
	Let $V_1$ and $V_2$ be weak divergences on the set $\zo^m$ and $f:\R
	\to \R$ be a function such that $V_1\diver P{U_M} \leq f\of{V_2\diver
	P{U_m}}$ for all distributions $P$ on $\zo^m$. Then if $f$ is
	increasing on $(0, \eps)$, every $(k,\eps)$ extractor $\Ext$ for $V_1$
	is also a $(k, f(\eps))$-extractor for $V_2$, and if $f$ is also
	concave, then if $\Ext$ is strong or average-case as a
	$V_1$-extractor, it has the same properties as a $(k, f(\eps))$
	extractor for $V_2$.
\end{lem}

Importantly, the KL divergence does not have the flaws of subgaussian
distance discussed in \cref{sec:composition}. The classic
\emph{data-processing inequality} says that KL divergence is
non-increasing under postprocessing by (possibly randomized) functions,
and the \emph{chain rule} for KL divergence says that
\[
	\KL\diver{A, B}{X, Y} = \KL\diver AX + \E[a\sim A]{\KL\diver{B|_{A=a}}{Y|_{X=a}}}
\]
for all distributions $A$, $B$, $X$, and $Y$, so that in particular
\[\E[s\sim U_d]{ \KL\diver{\Ext(X, s)}{U_m}} = \KL\diver{U_d, \Ext(X,
U_d)}{U_d, U_m}\]
and prepending the seed of a strong KL-extractor
does in fact
give a non-strong KL-extractor:
\begin{lem}\label{lem:kl_strong_prepend}
	A function $\Ext:\zo^n \times \zo^d \to \zo^m$ is a $(k,\eps)$ strong
	KL-extractor (respectively strong average-case KL-extractor) if
	and only if the function $\Ext':\zo^n \times \zo^d \to \zo^{d + m}$
	defined by $\Ext'(x, s) = (s, \Ext(x, s))$ is a (non-strong) $(k,
	\eps)$ KL-extractor (respectively average-case KL-extractor).
\end{lem}

Furthermore, KL divergence satisfies
a type of triangle inequality when combined with higher R\'enyi
divergences:
\begin{lem}\label{lem:kl_triangle}
	Let $P$, $Q$, and $R$ be distributions over a finite set $\X$. Then
	for all $\alpha > 0$, it holds that
	\[
		\KL\diver{P}{R} \leq \of{1 + \frac1\alpha}\cdot \KL\diver{P}{Q} +
		\D_{1+\alpha}\diver{Q}{R}
	\]
\end{lem}
\begin{proof}
	This follows from a characterization of R\'enyi divergence due to van
	Erven and Harremo\"es \cite[Lemma 6.6]{vanerven_when_2010} \cite[Theorem
	30]{van_har_renyi_2014} and Shayevitz \cite[Theorem
	1]{shayevitz_renyi_2011}, who prove that for for every positive real
	$\beta\neq 1$ and distributions $X$ and $Y$ that
	\[
		(1 - \beta) \D_\beta\diver{X}{Y} = \inf_Z\setd[\big]{\beta \KL\diver{Z}{X}
		+ (1 - \beta) \KL\diver{Z}{Y}}.
	\]
	In particular, choosing $\beta = 1 + \alpha$, $X = Q$, and $Y = R$ and
	upper bounding the infimum by the particular choice of $Z = P$ gives
	the claim.
\end{proof}

\subsection{Composition}
These properties imply that composition does work as we want (without
any loss depending on the output length $m$) assuming we have extractors
for KL and higher divergences.
\begin{thm}[Composition for high min-entropy R{\'e}nyi entropy extractors, cf.~\cite{gol_wig_tiny_1997}]\label{thm:divergence_composition}
	Suppose
	\begin{enumerate}
		\item $\Ext_{out}:\zo^n\times\zo^d \to\zo^m$ is an
	$(n-\log(1/\delta), \eps_{out})$ extractor for $\D_{1+\alpha}$ with $\alpha
	> 0$,
		\item $\Ext_{in}:\zo^{n'}\times \zo^{d'}\to \zo^{d}$ is an $(n' -
	\log(1/\delta), \eps_{in})$ average-case KL-extractor,
	\end{enumerate}
	and define $\Ext:\zo^{n + n'} \times \zo^{d'} \to \zo^m$ by
	$\Ext\of[\big]{(x, y), s} = \Ext_{out}(x, \Ext_{in}(y, s))$.
	Then $\Ext$ is an $\of{n + n' - \log(1/\delta),
	\eps_{out}+(1+1/\alpha)\cdot \eps_{in}}$ extractor for KL.
	Furthermore, if $\Ext_{in}$ is a strong average-case KL-extractor,
	then $\Ext$ is a strong KL-extractor, and if $\Ext_{out}$ is
	average-case then so is $\Ext$.
\end{thm}
\begin{proof}
	Let $(Z, X, Y)$ be jointly distributed random variables with $X$
	distributed over $\zo^n$ and $Y$ over $\zo^{n'}$ such that
	$\condminent(X, Y|Z) \geq n + n' - \log(1/\delta)$. Let $S'$ be a
	distribution over $\zo^{d'}$ which is independent of $X$, $Y$, and
	$Z$. Then for every $z\in \Supp(Z)$, we have by
	\cref{lem:kl_triangle} and the data-processing inequality for KL
	divergence that
	\begin{align*}
		&\KL\diver{\Ext\of{(X|_{Z=z},Y|_{Z=z}), S'}}{U_m}\\
		&\hspace{2em}=\KL\diver{\Ext_{out}\of{X|_{Z=z}, \Ext_{in}(Y|_{Z=z}, S')}}{U_m}\\
		&\hspace{2em}\leq \of{1+1/\alpha}\cdot \KL\diver{\Ext_{out}\of{X|_{Z=z}, \Ext_{in}(Y|_{Z=z}, S')}}{\Ext_{out}\of{X|_{Z=z}, U_d}}\\
		&\hspace{2em}\hphantom{\leq \of{1\vphantom{+1/\alpha}}}+ \D_{1+\alpha}\diver{\Ext_{out}\of{X|_{Z=z}, U_d}}{U_m}\\
		&\hspace{2em}\leq \of{1+1/\alpha}\cdot \KL\diver{X|_{Z=z}, \Ext_{in}(Y|_{Z=z}, S')}{X|_{Z=z}, U_{d}} + \D_{1+\alpha}\diver{\Ext_{out}\of{X|_{Z=z}, U_d}}{U_m}\\
		&\hspace{2em}= \of{1+1/\alpha}\cdot \E[x\sim X|_{Z=z}]{\KL\diver{\Ext_{in}(Y|_{X=x,Z=z}, S')}{U_{d}}} + \D_{1+\alpha}\diver{\Ext_{out}\of{X|_{Z=z}, U_d}}{U_m}
	\end{align*}
	where the last equality follows from the chain rule for
	KL divergence. Now by standard properties of conditional
	min-entropy (see for example \cite[Lemma
	2.2]{dod_ost_rey_smi_fuzzy_2008}), we know that $\condminent(X|Z) \geq
	\condminent(X,Y|Z) - \log\abs{\Supp(Y)} \geq n - \log(1/\delta)$ and
	$\condminent(Y|X,Z) \geq \condminent(X,Y|Z) - \log\abs{\Supp(X)} \geq
	n' - \log(1/\delta)$.

	If $\Ext_{out}$ is not average-case, take $Z$ to be a constant
	independent of $X$ and $Y$, and if $\Ext_{out}$ is average-case then
	take the average of both sides over $Z$. The claim for non-strong
	$\Ext_{in}$ then follows by taking $S' = U_d$ which bounds the first
	term by $(1+1/\alpha)\cdot \eps_{in}$ and the second by $\eps_{out}$.
	The claim for strong $\Ext_{in}$ follows by choosing $S' = U_{\set s}$
	to be the point mass on $s\in \zo^d$ and then taking the expectation
	of both sides over a uniform $s\in \zo^d$.
\end{proof}
\begin{rmrk}\label{rmrk:composition_entropy_loss}
	\Cref{thm:divergence_composition} in fact a construction of a
	\emph{block-source} KL-extractor, meaning that the claimed error
	bounds hold for any joint distributions $(X, Y)$ such that $\minent(Y)
	\geq n' - \log(1/\delta)$ and $\condminent(X|Y) \geq n -
	\log(1/\delta)$ rather than just those distributions with $\minent(X,
	Y) \geq n + n' - \log(1/\delta)$. The extra $\log(1/\delta)$ entropy
	loss inherent in the non-block analysis is why Reingold, Wigderson,
	and Vadhan \cite{rei_vad_wig_entropy_conference} introduced the
	zig-zag product for extractors, which we will apply for KL-extractors
	in \cref{cor:zigzag}.
\end{rmrk}

\subsection{Existing explicit constructions}

The construction of \cref{thm:divergence_composition} required both a
$\D_{1+\alpha}$-extractor and an average-case KL-extractor, so for
the result not to be vacuous we need to show the existence of such
extractors. Thankfully, \cref{defn:renyi_divergence} implies that
extractors for $\ell_2$ are also extractors for $\D_2$, so we can use
existing $\ell_2$ extractors from the literature, such as the Leftover
Hash Lemma of Impagliazzo, Levin, and Luby
\cite{imp_lev_lub_pseudorandom_1989} (see also
\cite{mcinnes_cryptography_1987,ben_bra_rob_privacy_1988}) and its
variant using almost-universal hash functions due to Srinivasan and
Zuckerman \cite{sri_zuc_computing_1999}.

\begin{prop}[\cite{mcinnes_cryptography_1987,ben_bra_rob_privacy_1988,%
	imp_lev_lub_pseudorandom_1989,imp_zuc_how_1989,%
	sri_zuc_computing_1999,dod_ost_rey_smi_fuzzy_2008}]\label{prop:lhl}
	Let $\mathcal{H}$ be a collection of $\eps$-almost universal hash
	functions from the set $\zo^n$ to the set $\zo^m$, meaning that for
	all $x\neq y\in\zo^n$ it holds that $\PR[h\sim \mathcal{H}]{h(x) =
	h(y)} \leq (1+\eps)/2^m$. Then the function $\Ext:\zo^n\times
	\mathcal{H}\to\mathcal{H}\times \zo^m$ defined by $\Ext(x, h) = (h,
	h(x))$ is an average-case $(m + \log(1/\eps), 2/\ln 2 \cdot \eps)$
	$\D_2$-extractor.

	In particular, for every $k,n\in \mathbb N$ and $1 >\eps > 0$ there is
	an explicit strong average-case $(k, \eps)$ extractor for $\D_2$ (and
	KL) with seed length $d = O(k + \log(n/\eps))$ and $m = k -
	\log(1/\eps) - O(1)$, given by $\Ext'(x, h) = h(x)$ for $h$ drawn from
	an appropriate almost-universal hash family.
\end{prop}
\begin{proof}
	The $\D_2$ claim is implicit in Rackoff's proof of the Leftover Hash
	Lemma (see \cite{imp_zuc_how_1989}) and Srinivasan and Zuckerman's
	proof of the claim for total variation \cite{sri_zuc_computing_1999},
	which both analyzed the \emph{collision probability} of the output,
	and the average-case claim was proved by Dodis, Ostrovsky, Reyzin, and
	Smith \cite{dod_ost_rey_smi_fuzzy_2008}, though we include a proof
	here for completeness.

	Given a joint distribution $(Z, X)$ such that $X$ is distributed over
	$\zo^n$ with $\condminent(X|Z)\geq m + \log(1/\eps)$, we have
	\begin{align*}
		&\E[z\sim Z]{\D_2\diver{\Ext(X|_{Z=z}, \mathcal{H})}{\mathcal{H} \times
			U_m}}\\
		&\hspace{3em}= \E[z\sim Z]{\log\of{2^m\cdot \abs{\mathcal H}\cdot \PR[h,h'\sim
			\mathcal{H}, x,x'\sim X|_{Z=z}]{(h, h(x)) = (h', h'(x'))}}}\\
		&\hspace{3em}= \E[z\sim Z]{\log\of{2^m\cdot \PR[h\sim
			\mathcal{H}, x,x'\sim X|_{Z=z}]{x = x' \vee \of[\big]{x\neq x' \wedge h(x) = h(x')}}}}\\
		&\hspace{3em}\leq \E[z\sim Z]{\log\of{2^m\cdot \of{2^{-\minent(X|_{Z=z})} + \frac{1+\eps}{2^m}}}}\\
		&\hspace{3em}\leq \log\of{\E[z\sim Z]{2^{m-\minent(X|_{Z=z})}} + 1+\eps}\tag{by Jensen's inequality}\\
		&\hspace{3em}=\log\of{2^{m-\condminent(X|Z)} + 1 + \eps} \leq \log\of{1+2\eps}\leq \frac{2}{\ln 2}\cdot \eps.
	\end{align*}
	The in particular statement follows from
	\cref{lem:d1pa-prepend-strong} below and from the existence of
	$\eps$-almost universal hash families with size $\poly(2^k, n,
	1/\eps)$ as constructed by \cite{sri_zuc_computing_1999}.
\end{proof}

To establish the claim about strong extractors, we generalize
\cref{lem:kl_strong_prepend} to extractors for $\D_{1+\alpha}$ for
$\alpha > 0$:
\begin{lem}\label{lem:d1pa-prepend-strong}
	If $\Ext:\zo^n\times\zo^d \to \zo^d \times \zo^m$ is a $(k, \eps)$
	$\D_{1+\alpha}$-extractor (respectively average-case
	$\D_{1+\alpha}$-extractor) for $\alpha > 0$ such that $\Ext(x, s) =
	(s, \Ext'(x, s))$, then $\Ext'$ is a strong $(k, \eps)$
	$\D_{1+\alpha}$-extractor (respectively strong average-case $(k,
	\eps)$ $\D_{1+\alpha}$-extractor).
\end{lem}
\begin{proof}
	\begin{align*}
		\E[s\sim U_d]{\D_{1+\alpha}\diver{\Ext'(X, s)}{U_m}}
		&= \E[s\sim U_d]{\frac1\alpha \log\of{1 + 2^{m\alpha} \sum_{y\in \zo^m} \PR{\Ext'(X, s) = y}^{1+\alpha}}}\\
		&\leq \frac1\alpha \log\of{1 + 2^{m\alpha} \E[s\sim U_d]{\sum_{y\in \zo^m} \PR{\Ext'(X, s) = y}^{1+\alpha}}}\\
		&= \frac1\alpha \log\of{1 + 2^{\alpha(m+d)}\sum_{(s, y)\in \zo^{d+m}} \PR{(U_d, \Ext'(X, U_d)) = (s, y)}^{1+\alpha}}\\
		&= \D_{1+\alpha}\diver{\Ext(X, U_d)}{U_d, U_m}\tag*{\qedhere}
	\end{align*}
\end{proof}

Following Vadhan \cite{vadhan_pseudorandomness_2012}, we also note that
the extractor based on expander walks due to Goldreich and Wigderson
\cite{gol_wig_tiny_1997}, which has the nice property that its seed
length depends only on $n - k$ the entropy deficiency of the source
rather than $n$ itself, is also an $\ell_2$ extractor. Before stating
the extractor formally, we introduce some notation and terminology
we will need.

\begin{defn}
	Let $G$ be a $D$-regular graph on $\zo^n$ with adjacency matrix $A_G$
	and transition matrix $M_G = \frac1D A_G$. Then if $M_G$ has
	eigenvalues $1 = \lambda_1\geq \lambda_2\geq\cdots\geq \lambda_n\geq
	-1$, the \emph{spectral expansion} of $G$ is $\lambda =
	\max\set{\lambda_2, -\lambda_n}$. A function $\Gamma_G:\zo^n\times [D]
	\to \zo^n$ is a \emph{neighbor function of $G$} if there is some
	labelling of the edges of $G$ for which $\Gamma_G(v, i)$ is the vertex
	obtained by following the $i$th edge out of $v$ in $G$. $\Gamma_G$ is
	\emph{consistently labelled} if for all $v\neq v'\in\zo^n$ and $i\in
	[D]$ we have $\Gamma(v, i)\neq \Gamma(v', i)$, that is, at most one
	incoming edge is labelled by $i$.
\end{defn}

\begin{lem}\label{lem:expander_is_l2_ext}
	Let $\Gamma:\zo^n\times\zo^d\to\zo^n$ be the neighbor function of
	a graph $G$ with spectral expansion $\lambda$. Then for every $0\leq k
	\leq n$, $\Gamma$ is a $\of{k, \lambda\sqrt{2^{-k}-2^{-n}}}$
	$\ell_2$-extractor and a $\of{k, \log\of{1 +
	\lambda^2\of{2^{n-k}-1}}}$ $\D_2$-extractor. Furthermore, if
	$\Gamma_G$ is consistently labelled, then the function $\Waste(x,
	s)=s$ is such that $(\Gamma_G, \Waste)$ is an injection out of
	$\zo^n\times\zo^d$.

	In particular, if $\lambda^2 \leq \eps \cdot 2^{k-n}$ then $\Ext$ is
	an average-case $(k, \sqrt{\eps} \cdot 2^{-n/2})$ $\ell_2$-extractor
	and an average-case $(k, \eps/\ln 2)$ $\D_2$-extractor.
\end{lem}
\begin{proof}
	If $X$ is a distribution over $\zo^n$ with $\minent(X) \geq k$, then
	$\log\of{1 + 2^n d_{\ell_2}\of{X, U_n}} = \D_2\diver X{U_n} \leq
	\D_\infty\diver X{U_n} \leq n - k$ so that $d_{\ell_2}\of{P, U_n} \leq
	\sqrt{2^{-k} - 2^{-n}}$. The $\ell_2$-extractor result follows since
	the action of $\Gamma_G$ reduces the $\ell_2$ distance to uniform by a
	factor of $\lambda$, and the $\D_2$-extractor result from the fact
	that $D_2\diver P{U_n} = \log\of{1 + 2^n d_{\ell 2}\of{P, U_n}^2}$ for
	every distribution $P$ on $\zo^n$.

	For the furthermore claim, we need to show that $\of{x, s}\mapsto
	\of{\Gamma_G(x, s), s}$ is an injection, or equivalently that given
	$\Gamma_G(x, s)$ and $s$, one can recover $x$. But by definition of
	consistent labelling, at most one edge into $\Gamma_G(x, s)$ is
	labelled by $s$, and so taking this edge from $\Gamma_G(x, s)$ gives
	$x$, as desired.  Finally, the in particular claim follows by Jensen's
	inequality, since $\log$ and square-root are concave, and $\E[z\sim
	Z]{2^{-\minent(X|_{Z=z})}} = 2^{-\condminent(X|Z)}$ by definition.
\end{proof}
\begin{rmrk}
	The fact that $(s, \Ext(x, s))$ is an injection implies that, unlike
	for the extractors from hashing of \cref{prop:lhl}, the result of
	prepending the seed to the output of the expander-walk extractor does
	\emph{not} give a $\D_2$ extractor. However, it will be very useful in
	concert with Reingold, Vadhan, and Wigderson's zig-zag product for
	extractors \cite{rei_vad_wig_entropy_conference} to avoid the entropy
	loss in \cref{thm:divergence_composition}.
\end{rmrk}

\begin{cor}[\cite{gol_wig_tiny_1997} {\cite[Discussion after Theorem
	6.22]{vadhan_pseudorandomness_2012}}]\label{cor:expander_extractor}
	There is a universal constant $C \geq 1$ such that for every
	$1>\eps>0$, $\Delta > 0$, and $n\in\mathbb N$ there is an explicit $(n
	- \Delta, \eps/\ln 2)$ average-case $\D_2$-extractor (respectively $(n -
	\Delta, \sqrt{\eps}\cdot 2^{-n/2})$ average-case $\ell_2$-extractor)
	$\Ext:\zo^n \times \zo^d \to \zo^n$ with $d=\ceil{C\cdot (\Delta +
	\log(1/\eps))} + O(1)$ such that the function $(x, s) \mapsto (s,
	\Ext(x, s))$ is an injection.

	Moreover, if there is an explicit construction of consistently
	labelled neighbor functions for Ramanujan graphs over $\zo^n$ with
	degree $D = O\of{2^{\Delta}/\eps}$, then one can take $C = 1$.
\end{cor}
\begin{proof}
	By \cref{lem:expander_is_l2_ext} it suffices to demonstrate the
	existence of an explicit $D$-regular expander graph over $\zo^n$ with
	a consistently labelled neighbor function $\Gamma_G$, spectral
	expansion $\lambda^2 \leq \eps \cdot 2^{-\Delta}$, and $D =
	O\of{\of{2^{\Delta}/\eps}^C}$.  The claim about Ramanujan graphs is
	thus immediate since a Ramanujan graph with degree $O(2^\Delta/\eps)$
	has $\lambda^2 \leq 4/D \leq \eps\cdot 2^{-\Delta}$.

	Without the assumption of good Ramanujan graphs, we can use a power of
	the the explicit constant degree expander of Margulis--Gabber--Galil
	\cite{margulis_explicit_1973,gab_gal_explicit_1981} (technically this
	requires $n$ even, which following Goldreich
	\cite{goldreich_basic_2011} we can fix when $n$ is odd by joining two
	graphs on $\zo^{n-1}$ by the canonical perfect matching, and we can
	add self-loops to ensure the degree is a power of $2$). This graph $G$
	is consistently labelled with degree $D_{MGG} = O(1)$ and constant
	spectral expansion $\lambda_{MGG} < 1$. Then the graph $G^w$ on
	$\zo^n$ with edges representing $w$-length paths has spectral
	expansion $\lambda_{MGG}^w$ and degree $D_{MGG}^w$, which for $w =
	\ceil{\log_{\lambda_{MGG}}(1/2) \cdot \of{\Delta + \log(1/\eps)}}$
	gives $\lambda \leq \eps\cdot 2^{-\Delta}$ and degree $D =
	O\of{\of{2^{\Delta}/\eps}^C}$ for $C \leq \log\of{D_{MGG}}\cdot
	\log_{\lambda_{MGG}}(1/2)$ as desired.
\end{proof}

We argued that the above extractors are KL-extractors using the fact
they are $\ell_2$ (and thus $\D_2)$ extractors, but one can also show
that any total variation extractor with sufficiently small error is a
KL-extractor, albeit with some loss of parameters.
\begin{lem}\label{lem:tv_extractor_is_kl}
	For every $(k, \eps)$ extractor $\Ext:\zo^n \times \zo^d \to \zo^m$ for
	total variation distance such that $\eps \leq 1/2$, $\Ext$ is also a $(k,
	m\cdot \eps + h(\eps))$-KL-extractor, where $h(x) = x\log(1/x) +
	(1-x)\log(1/(1-x))$ is the binary entropy function. Furthermore, if $\Ext$
	is strong, average-case, or both as a total variation extractor, then
	it has the same properties as a KL-extractor.

	In particular, if $\eps' = \frac{\min(\eps, 1/2)}{48(m +
	\log(1/\eps))}$, then every $(k, \eps')$ extractor (respectively
	strong extractor) is an average-case $(k, \eps)$ KL-extractor
	(respectively strong average-case $(k, \eps)$ KL-extractor).
\end{lem}
\begin{proof}
	The main claim is an immediate corollary of
	\cref{lem:dg_vs_kl,lem:divergence_bounds}. The in particular statement
	follows since $\Ext$ being a $(k, \eps')$ extractor (respectively
	strong extractor) implies by \cref{thm:func_ext_are_average} that it
	is a $(k, 3\eps')$ average-case (respectively strong average-case)
	extractor, so since we have chosen $\eps'$ to make $m\cdot 3\eps' +
	h(3\eps') \leq \eps$, we know $\Ext$ is an average-case $(k, \eps)$
	KL-extractor (respectively strong average-case KL-extractor).
\end{proof}
\begin{rmrk}
	Reducing $\eps$ by a factor of $m + \log(1/\eps)$ increases the seed
	length and entropy loss of the input extractor. For the former, this
	is often (but not always) tolerable since the input extractor may
	already depend suboptimally on $\log(n/\eps)$. For the latter, we will
	show in \cref{cor:rrv_entropy_loss} how to use the transform of Raz,
	Reingold, and Vadhan \cite{raz_rei_vad_extracting_2002} to recover
	$O(\log(m/\eps))$ bits of lost entropy (at least this much must be
	lost by Radhakrishnan and Ta-Shma \cite{rad_ta-_bounds_2000}) at a
	cost of $O(\log(n/\eps))$ in the seed length.
\end{rmrk}

Instantiating \cref{lem:tv_extractor_is_kl} with the
Guruswami--Umans--Vadhan \cite{gur_uma_vad_unbalanced_2009} extractor
for total variation distance, we see that the increased seed length and
entropy loss are simply absorbed into the existing hidden constants:
\begin{thm}[KL-analogue of {\cite[Theorem 1.5]{gur_uma_vad_unbalanced_2009}}]\label{cor:kl_guv}
	For every $n\in \mathbb N$, $k\leq n$, and $1>\alpha,\eps > 0$, there
	is an explicit average-case (respectively strong average-case)
	$(k,\eps)$ KL-extractor $\Ext:\zo^n\times \zo^d \to \zo^m$ with $d
	\leq \lg n + O_{\alpha}(\lg(k/\eps))$ and $m \geq (1-\alpha) k$
	(respectively $m \geq (1-\alpha)k - O_\alpha(\log(n/\eps))$).
\end{thm}

\subsection{Reducing the entropy loss of KL-extractors}

In this section, we show how to avoid the entropy loss inherent in
\cref{thm:divergence_composition} using the zig-zag product for
extractors, introduced by Reingold, Vadhan, and Wigderson
\cite{rei_vad_wig_entropy_conference}.  This product combines a
technique of Raz and Reingold \cite{raz_rei_recycling_1999} to preserve
entropy and the method of Wigderson and Zuckerman
\cite{wig_zuc_expanders_1999} to extract entropy left over in a source
after an initial extraction, and we show that these techniques extend to
the setting of KL-extractors.  Furthermore, these techniques (along
with the Leftover Hash Lemma) are also the key to the transformation
of Raz, Reingold, and Vadhan \cite{raz_rei_vad_extracting_2002} to
convert an arbitrary extractor into one with optimal entropy loss, so we
show that this transformation works for KL-extractors as well.

For all of these results, the key is the following lemma:
\begin{lem}[Re-extraction from leftovers]\label{lem:weak_waste_reextraction}
	Let
	\begin{enumerate}
		\item $\Ext_1:\zo^n \times \zo^{d_1} \to \zo^{m_1}$ be a $(k_1,
			\eps_1)$ KL-extractor,
		\item $\Waste_1:\zo^n \times \zo^{d_1} \to
			\zo^{w}$ be a function such that $(\Ext_1, \Waste_1):\zo^n \times
			\zo^{d_1} \to \zo^{m_1}\times \zo^w$ is an injective map,
		\item $\Ext_2:\zo^w \times \zo^{d_2} \to \zo^{m_2}$ be a $(k_2, \eps_2)$
	average-case KL-extractor for $k_2 \leq k_1 + d_1 - m_1$.
	\end{enumerate}
	Then $\Ext:\zo^n \times \zo^{d_1 + d_2} \to \zo^{m_1 + m_2}$ defined
	by $\Ext\of[\big]{x, (s, t)} = \of{\Ext_1(x, s),
	\Ext_2\of[\big]{\Waste_1(x, s), t}}$ is a $(k_1, \eps_1 + \eps_2)$
	KL-extractor. Furthermore, if $\Ext_1$ is average-case then so is
	$\Ext$.
\end{lem}
\begin{rmrk}
	The pair $(\Ext_1, \Waste_1)$ is a special case of what Raz and
	Reingold \cite{raz_rei_recycling_1999} called an
	\emph{extractor-condenser pair}. One can think of $\Waste_1$ as
	preserving ``leftovers'' or ``waste,'' which is then ``re-extracted''
	or ``recycled'' by $\Ext_2$. The identity function on $\zo^n \times
	\zo^{d_1}$ is a valid choice of $\Waste_1$, but the advantage of the
	more general formulation is that $w$ can be much smaller than $n +
	d_1$, and most known explicit constructions of extractors have seed
	length depending on the input length of the source.
\end{rmrk}
\begin{proof}
	Given any joint distribution $(Z, X)$ such that $X$ is distributed
	over $\zo^n$ and $\condminent(X|Z)\geq k_1$, we have for every $z\in
	\Supp(Z)$ that
	\begin{align}
		&\KL\diver{\Ext(X|_{Z=z},\of{U_{d_1},U_{d_2}})}{U_{m_1+m_2}}\nonumber\\
		&\hspace{3em}=\KL\diver{\Ext_1(X|_{Z=z},U_{d_1}), \Ext_2\of[\big]{\Waste_1\of{X|_{Z=z}, U_{d_1}}, U_{d_2}}}{U_{m_1},U_{m_2}}\nonumber\\
		&\hspace{3em}=\KL\diver{\Ext_1\of{X|_{Z=z},U_{d_1}}}{U_{m_1}}\nonumber\\
		&\hspace{3em}\qquad+ \E[o_1 \sim \Ext_1(X|_{Z=z},s)]{\KL\diver{\Ext_2\of[\Big]{\Waste_1\of{X,U_{d_1}}|_{Z=z,\Ext_1(X,U_{d_1})=o_1}, U_{d_2}}}{U_{m_2}}}\label{eqn:reext-err}
	\end{align}
	where the last line follows from the chain rule for KL divergence.
	Note that
	\begin{align*}
		&\condminent\of[\Big]{\Waste_1\of[\big]{X, U_{d_1}} \:\Big\vert\:
		Z, \Ext_1\of[\big]{X, U_{d_1}}}\\
		&\hspace{5em}=\condminent\of[\Big]{\Ext_1\of[\big]{X, U_{d_1}}, \Waste_1\of[\big]{X, U_{d_1}} \:\Big\vert\:
Z, \Ext_1\of[\big]{X, U_{d_1}}}\\
		&\hspace{5em}=
		\condminent\of[\Big]{X, U_{d_1} \:\Big\vert\:
		Z, \Ext_1\of[\big]{X, U_{d_1}}}\tag{$(\Ext_1, \Waste_1)$ is an injection}\\
		&\hspace{5em}\geq
		\condminent\of{X, U_{d_1}\:|\:Z} - \log\abs{\Supp\of{\Ext_1\of[\big]{X, U_{d_1}}}}\tag{*}\\
		&\hspace{5em}=
		\condminent\of{X\:|\:Z} + \minent(U_{d_1}) - \log\abs{\Supp\of{\Ext_1\of[\big]{X, U_{d_1}}}}\tag{by independence}\\
		&\hspace{5em}\geq k_1 + d_1 - m_1 \geq k_2
	\end{align*}
	where the line (*) follows from standard properties of conditional
	min-entropy (e.g.~\cite[Lemma 2.2]{dod_ost_rey_smi_fuzzy_2008}). That
	$\Ext$ is a $(k_1, \eps_1 + \eps_2)$ KL-extractor now follows
	immediately from \cref{eqn:reext-err} by taking $Z$ independent of
	$X$, and the average-case claim follows from taking expectations over
	$z\sim Z$.
\end{proof}
\begin{rmrk}
	The proof above in fact works any weak divergence $\D$ such that
	$\D\diver{X, Y}{U_{m_1}, U_{m_2}} \leq \D\diver X{U_{m_1}} + \E[x\sim
	X]{\D\diver{Y|_{X=x}}{U_{m_2}}}$ for all joint distributions $(X, Y)$
	independent of $(U_{m_1}, U_{m_2})$.  In particular, the proof also
	gives \cref{lem:weak_waste_reextraction} for standard (total
	variation) extractors.
\end{rmrk}
By \cref{lem:kl_strong_prepend}, we get an analogous result for strong
KL-extractors.
\begin{cor}\label{cor:strong_waste_reextraction}
	Let
	\begin{enumerate}
		\item $\Ext_1:\zo^n \times \zo^{d_1} \to \zo^{m_1}$ be a strong $(k_1,
	\eps_1)$ KL-extractor,
		\item $\Waste_1:\zo^n \times \zo^{d_1} \to
	\zo^{w}$ be a function such that the map $(x, s)\mapsto (s, \Ext_1(x,
	s), \Waste_1(x, s))$ is an injection,
		\item $\Ext_2:\zo^w \times
	\zo^{d_2} \to \zo^{m_2}$ be a $(k_2, \eps_2)$ strong average-case
	KL-extractor for $k_2 \leq k_1 - m_1$.
	\end{enumerate}
	Then $\Ext:\zo^n \times \zo^{d_1 + d_2} \to \zo^{m_1 + m_2}$ defined
	by $\Ext\of[\big]{x, (s, t)} = \of{\Ext_1(x, s),
	\Ext_2\of[\big]{\Waste_1(x, s), t}}$ is a strong $(k_1, \eps_1 +
	\eps_2)$ KL-extractor. Furthermore, if $\Ext_1$ is average-case then
	so is $\Ext$.
\end{cor}

The zig-zag product for extractors due to Reingold, Vadhan, and
Wigderson \cite{rei_vad_wig_entropy_conference} (in the special case of
injective $(\Ext, \Waste)$-pairs) is an immediate consequence of
\cref{lem:weak_waste_reextraction} and \cref{thm:divergence_composition}
our basic composition result. Recall that
\cref{thm:divergence_composition} was able to combine an ``outer''
extractor, generally taken to have seed length depending only (but
linearly) on $n - k$, with an ``inner'' extractor to produce seeds for
the outer extractor with logarithmic seed length.  However, as discussed
in \cref{rmrk:composition_entropy_loss} that basic composition
necessarily lost $\log(1/\delta)$ bits of entropy, so the zig-zag
product uses \cref{lem:weak_waste_reextraction} to recover this entropy,
using an $(\Ext, \Waste)$-pair to ensure that the re-extraction adds
additional seed length depending logarithmically on $n-k$ rather than
$n$.

\begin{cor}[Zig-zag product for KL-extractors, analogous to
	{\cite[Theorem
	3.6]{rei_vad_wig_entropy_conference}}]\label{cor:zigzag}
	Let
	\begin{enumerate}
		\item $\Ext_{out}:\zo^n\times\zo^d \to\zo^m$ be an
			$(n-\log(1/\delta), \eps_{out})$ extractor for $\D_{1+\alpha}$
			with $\alpha > 0$,
		\item $\Waste_{out}:\zo^n \times \zo^d\to \zo^{w}$ be a function
			such that the pair $(\Ext_{out}, \Waste_{out})$ is an injection from
			$\zo^n\times \zo^d$,
		\item $\Ext_{in}:\zo^{n'}\times \zo^{d'}\to \zo^{d}$ be an $(n' -
			\log(1/\delta), \eps_{in})$ average-case KL-extractor,
		\item $\Waste_{in}:\zo^{n'} \times \zo^{d'} \to \zo^{w'}$ be such
			that the pair $(\Ext_{in}, \Waste_{in})$ is an injection from
			$\zo^{n'}\times \zo^{d'}$,
		\item $\Ext_{waste}: \zo^{w + w'} \times \zo^{d''}
			\to \zo^{m''}$ be an average-case $(n+n'-\log(1/\delta)-m,
			\eps_{waste})$ KL-extractor,
	\end{enumerate}
	and define
	\begin{enumerate}
		\item $\Ext_{comp}:\zo^{n + n'} \times \zo^{d'} \to \zo^{m}$ by
			$\Ext_{comp}\of[\big]{(x, y), s} = \Ext_{out}\of[\big]{x,
			\Ext_{in}(y, s)}$ as in \cref{thm:divergence_composition},
		\item $\Waste_{comp}:\zo^{n+n'}\times \zo^{d'} \to \zo^{w + w'}$ by
			$\Waste_{comp}\of{(x, y), s} = \of{\Waste_{out}(x, \Ext_{in}(y,
			s)), \Waste_{in}(y, s)}$,
		\item $\Ext:\zo^{n + n'} \times \zo^{d'+d''} \to \zo^{m+m''}$ by
			\[\Ext\of[\big]{(x, y), (s, t)} =
			\of[\Bigg]{\Ext_{comp}\of[\big]{(x, y), s},
			\Ext_{waste}\of[\Big]{\Waste_{comp}\of[\big]{(x, y), s}, t}}\]
			as in \cref{lem:weak_waste_reextraction}.
	\end{enumerate}
	Then $\Ext$ is an $\of{n + n' - \log(1/\delta),
	\eps_{out}+(1+1/\alpha)\cdot \eps_{in} + \eps_{waste}}$-extractor for
	KL. Furthermore, if $\Ext_{in}$ and $\Ext_{waste}$ are strong
	average-case KL-extractors, then $\Ext$ is a strong KL-extractor, and
	if $\Ext_{out}$ is average-case then so is $\Ext$.
\end{cor}
\begin{proof}
	We claim that $\Waste_{comp}$ is such that $(\Ext_{comp},
	\Waste_{comp})$ is an injection: by assumption on $(\Ext_{out},
	\Waste_{out})$ we have that given $\Ext_{out}(x, \Ext_{in}(y, s))$ and
	$\Waste_{out}(x, \Ext_{in}(y, s))$ we can recover $x$ and
	$\Ext_{in}(y, s)$, and by assumption on $(\Ext_{in}, \Waste_{in})$
	given $\Ext_{in}(y, s)$ and $\Waste_{in}(y, s)$ we can recover $(y,
	s)$, so that $(\Ext_{comp}, \Waste_{comp})$ has an inverse and is
	injective as desired.  Therefore, since
	\cref{thm:divergence_composition} implies $\Ext_{comp}$ is an $(n + n'
	- \log(1/\delta), \eps_{out} + (1 + 1/\alpha)\cdot \eps_{in})$
	KL-extractor, the result follows from
	\cref{lem:weak_waste_reextraction}. The furthermore claims follow from
	the corresponding claims of these lemmas (and
	\cref{cor:strong_waste_reextraction} for the strong case).
\end{proof}
\begin{rmrk}
	\Cref{cor:zigzag} was presented by Reingold, Vadhan, and Wigderson
	\cite{rei_vad_wig_entropy_conference} as a transformation that
	combined three extractor-condenser pairs into a new
	extractor-condenser pair.  We do not use this generality, so for
	simplicity we do not present it here, but both
	\cref{lem:weak_waste_reextraction} and \cref{cor:zigzag} can be easily
	extended in this manner if required.
\end{rmrk}

The Raz--Reingold--Vadhan \cite{raz_rei_vad_extracting_2002}
transformation to avoid entropy loss follows similarly using the
Leftover Hash Lemma (\cref{prop:lhl}).
\begin{cor}[KL-extractor analogue of {\cite[Lemma 28]{raz_rei_vad_extracting_2002}}]\label{cor:rrv_entropy_loss}
	Let $\Ext_1:\zo^n \times \zo^{d_1} \to \zo^{m_1}$ be a strong $(k,
	\eps/2)$ KL-extractor with entropy loss $\Delta_1$, meaning $m_1 =
	k - \Delta_1$. Then for every $d_{extra} \leq \Delta_1$ there is an
	explicit $(k, \eps)$ strong KL-extractor $\Ext:\zo^n \times
	\zo^{d'} \to \zo^{m'}$ with seed length $d' = d_1 + O(d_{extra} +
	\log(n/\eps))$ and entropy loss $\Delta_1 - d_{extra} + \log(1/\eps) -
	O(1)$, meaning $m' = k - (\Delta_1 - d_{extra}) - \log(1/\eps) +
	O(1)$, which is computable in polynomial time making one oracle call
	to $\Ext_1$.  Furthermore, if $\Ext_1$ is average-case then so is
	$\Ext$.

	In particular, by taking $d_{extra} = \Delta_1$ we get an extractor
	with optimal entropy loss $\log(1/\eps) + O(1)$ by paying an
	additional $O(\Delta + \log(n/\eps))$ in seed length.
\end{cor}
\begin{proof}
	Let $\Waste_1:\zo^n \times \zo^{d_1} \to \zo^n$ be given by
	$\Waste_1(x, s) = x$, and let $\Ext_2:\zo^n \times \zo^{d_2} \to
	\zo^{m_2}$ be the strong average-case $(d_{extra}, \eps/2)$
	KL-extractor of \cref{prop:lhl} using almost-universal hash
	functions, so that $d_2 = O(d_{extra} + \log (n/\eps))$ and $m_2 =
	d_{extra} - \log(1/\eps) - O(1)$. The result follows from taking
	$\Ext$ to be the extractor of \cref{cor:strong_waste_reextraction}.
\end{proof}
\begin{rmrk}
	An analogous versions of the above claim for non-strong
	KL-extractors follows by taking $\Waste_1(x, s) = (x, s)$ and using
	\cref{lem:weak_waste_reextraction}.
\end{rmrk}

We can apply \cref{cor:rrv_entropy_loss} to \cref{cor:kl_guv} the KL-extractors from the
total variation extractors of Guruswami, Umans, and Vadhan
\cite{gur_uma_vad_unbalanced_2009}, thereby avoiding the extra
$O(\log(n/\eps))$ entropy loss in the strong extractors.
\begin{cor}\label{cor:kl_guv_lessloss}
	For every $n\in \mathbb N$, $1>\alpha,\eps > 0$, and $k,k'\geq0$ with
	$k + k' \leq n$, there is an explicit strong average-case
	$(k + k',\eps)$ KL-extractor $\Ext:\zo^n\times \zo^d \to
	\zo^m$ with $d \leq O_{\alpha}(\log(n/\eps)) + O(k')$ and
	$m \geq (1-\alpha)k + k' - \log(1/\eps) - O(1)$.
\end{cor}

\subsection{Lower bounds}
\label{sec:lb}

In this section, we give lower bounds on extractors for the R\'enyi
divergences $D_\beta$ of all orders, including the special case $\beta =
1$ of KL-extractors. A reader primarily interested in explicit
constructions of subgaussian samplers can skip to
\cref{sec:subgaussian_samplers}.

For R\'enyi divergences $D_\beta$ with $\beta \leq 1$ we reduce to
Radhakrishnan and Ta-Shma's \cite{rad_ta-_bounds_2000} lower bounds for
total variation extractors and \emph{dispersers}, which can be
understood as a one-sided relaxation of total variation extractors.

\begin{defn}[Sipser \cite{sipser_expanders_1988}, Cohen and Wigderson \cite{coh_wig_dispersers_1989}]
	\label{defn:disperser}
	A function $\Disp:\zo^n\times \zo^d\to\zo^m$ is a \emph{$(k,\eps)$
	disperser} if for all random variables $X$ over $\zo^n$ with
	$\minent(X)\geq k$, it holds that $\abs{\Supp(\Disp(X, U_d))} \geq
	(1-\eps)2^m$.
\end{defn}
Dispersers are of interest in the context of R\'enyi extractors because
the R\'enyi $0$-entropy of a random variable is the logarithm of its
support size (see \cref{defn:renyi_divergence}), and hence dispersers
are equivalent to $\D_0$-extractors:
\begin{lem}\label{lem:disperser_d0}
	$\Disp$ is a $(k,\eps)$ disperser if and only if $\Disp$ is a
	$\of{k, \log\of[\big]{1/(1-\eps)}}$ $\D_0$-extractor.
\end{lem}

Given \cref{lem:disperser_d0}, we can use the Radhakrishnan and Ta-Shma
\cite{rad_ta-_bounds_2000} lower bounds to give an optimal lower bound
on the seed length of $D_{\beta}$-extractors for $\beta\leq 1$ in terms
of the error $\eps$, input length $n$ and supported entropy $k$ (we will
give a matching non-explicit upper bound in the next section), as well
as lower bounds on the entropy loss. For the case $\beta = 1$ of
KL-extractors, the non-explicit upper bound (\cref{thm:kl_nonexplicit})
also shows that the entropy loss lower bound is optimal.
\begin{thm}\label{prop:small_renyi_lb}
	Let $0 \leq \beta \leq 1$ and $\Ext:\zo^n \times \zo^d \to \zo^m$ be a
	$(k, \eps)$ extractor for $D_\beta$ with $k\leq n - 2$, $d\leq m - 1$,
	and $2^{2-m} < \eps < 1/4$. Then $d\geq \log(n-k) + \log(1/\eps) -
	O(1)$ and $m \leq k + d - \log\log(1/\eps) + O(1)$. Furthermore, if
	$\eps$ is at most $\beta/(2\ln 2)$ then $m \leq k + d - \log(1/\eps) +
	\log(1/\beta) + O(1)$.
\end{thm}
\begin{proof}
	Since $D_\beta$ is nondecreasing in $\beta$ we have that $\Ext$ is a
	$(k, \eps)$ extractor for $\D_0$, and thus by \cref{lem:disperser_d0}
	it is a $(k, 1 - 2^{-\eps})$ disperser. Then the disperser seed length
	lower bound of Radhakrishnan and Ta-Shma \cite{rad_ta-_bounds_2000}
	tells us that $d\geq \log(n - k) + \log(1/(1-2^{-\eps})) - O(1) \geq
	\log(n-k) + \log(1/\eps) - O(1)$ and $m \leq k + d -
	\log\log(1/(1-2^{-\eps})) + O(1) \leq k + d - \log\log(1/\eps) +
	O(1)$.

	For the other entropy loss lower bound, we use Gilardoni's
	\cite{gilardoni_pinsker_2010} generalization of Pinsker's inequality,
	which shows in particular that $\tv(P, U_m) \leq \sqrt{\ln 2/(2\beta)
	\cdot \D_\beta\diver P{U_m}}$. Thus, $\Ext$ is also a $(k,
	\sqrt{\eps\cdot \ln 2/(2\beta)})$ total variation extractor, and if
	$\sqrt{\eps\cdot \ln 2/(2\beta)} \leq 1/2$ (equivalently $\eps \leq
	\beta/(2\ln 2)$) then the \cite{rad_ta-_bounds_2000} total variation
	extractor entropy loss lower bound implies that $m\leq k + d -
	2\log(1/\sqrt{\eps\cdot \ln 2/(2\beta)}) + O(1) \leq k + d -
	\log(1/\eps) + \log(1/\beta) + O(1)$.
\end{proof}
\begin{rmrk}
	For the case of $0 < \beta < 1$, we do not know whether the entropy
	loss lower bound of \cref{prop:small_renyi_lb} is tight.
\end{rmrk}

It is well-known that $\ell_2$-extractors (which are equivalent to
$\D_2$-extractors by \cref{defn:renyi_divergence}) require seed length
at least linear in $\min(n-k, m)$ (see e.g.~\cite[Problem
6.4]{vadhan_pseudorandomness_2012}). We generalize this to give a linear
seed length lower bound on $D_\beta$ extractors for all $\beta > 1$, in
the regime of constant $\eps$, improving on the logarithmic lower bound
given by \cref{prop:small_renyi_lb}.
\begin{thm}\label{prop:renyi_lb}
	Let $\Ext:\zo^n \times \zo^d \to \zo^m$ be a $(k, 0.99)$
	$\D_{1+\alpha}$-extractor for $\alpha > 0$.  Then $d \geq
	\min\setd[\big]{(n - k - 3)\cdot \alpha, (m-2)\cdot \alpha/(\alpha +
	1)}$.
\end{thm}
\begin{proof}
	We follow the strategy suggested by Vadhan \cite[Problem
	6.4]{vadhan_pseudorandomness_2012}, and view $\Ext$ as a bipartite
	graph with $N = \zo^n$ left-vertices, $M = \zo^m$ right-vertices, and
	$D = 2^d$ edges per left-vertex given by $E = \setd{(x\in \zo^n, y\in
		\zo^m) \given \exists s\in \zo^d: \Ext(x, s) = y}$.

	Assume for the sake of contradiction that $d \leq \alpha/(\alpha + 1)
	\cdot (m - 2)$ and $d\leq \alpha(n - k - 3)$, so that $M \geq
	4D^{1+1/\alpha}$ and $N/(8D^{1/\alpha}) \geq K$.
	Now, we claim there exists a set $T\subseteq \zo^m$ of size at most
	$M/(2D^{1+1/\alpha})$ such that $X = \setd{x\in \zo^n \given \exists
	s\in \zo^d \text{ s.t }\Ext(x, s)\in T}$ has size at least
	$N/(8D^{1/\alpha})\geq K$. This follows from the following iterative
	procedure: until $\abs{X} \geq N/(8D^{1/\alpha})$, choose the vertex
	$y\in \zo^m$ of highest degree, add it to $T$, and remove $y$ and its
	neighbors from the graph (the neighbors go in $X$). Then at each step
	we will add to $X$ a number of vertices at least the average degree
	\[
		\frac{(N - \abs{X})\cdot D}{M - \abs{T}}
		\geq \frac{(N - N/(8D^{1/\alpha}))\cdot D}{M}
		\geq \frac{ND}{2M},
	\]
	so that the size of $T$ will be at most $\ceil{N/(8D^{1/\alpha}) \cdot
	2M/ND} = \ceil{M/(4D^{1+1/\alpha})} \leq M/(2D^{1+1/\alpha})$ as
	desired. Now, since $X$ has size at least $K$ and $\Ext$ is a $(k,
	0.99)$ $\D_{1+\alpha}$-extractor, we have that
	\begin{align*}
		0.99
		&\geq \D_{1+\alpha}\diver{\Ext(U_X, U_d)}{U_m}\\
		&= \frac{1}{\alpha}\log\of{\sum_{y\in \zo^m} \frac{\PR{\Ext(U_X, U_D) = y}^{1+\alpha}}{2^{-m\alpha}}}\\
		&\geq \frac{1}{\alpha}\log\of{M^{\alpha}\sum_{y\in T} \PR{\Ext(U_X, U_D) = y}^{1+\alpha}}\\
		&\geq \frac{1}{\alpha}\log\of{M^{\alpha}\cdot \abs{T}^{-\alpha}\cdot \of{\sum_{y\in T}\PR{\Ext(U_X, U_d) = y}}^{1+\alpha}}\tag{By H\"older's inequality}\\
		&\geq \frac{1}{\alpha}\log\of{M^{\alpha} \cdot (M/(2D^{1+1/\alpha}))^{-\alpha} \cdot (1/D)^{1+\alpha}}
		= 1
		\tag{By definition of $T$}
	\end{align*}
	which is a contradiction, as desired.
\end{proof}
We can also use this lower bound to get a similar lower bound for
$d_{\ell_{1+\alpha}}$-extractors for all $\alpha > 0$, though in this
case the lower bound applies up to an error threshold that depends on
$\alpha$.
\begin{cor}\label{cor:l1pa-lb}
	Let $\Ext:\zo^n \times \zo^d \to \zo^m$ be a $\of{k, \eps_{\alpha} \cdot
	2^{-m\alpha/(1+\alpha)}}$ extractor for $d_{\ell_{1+\alpha}}$ where
	$\alpha > 0$ and $\eps_\alpha = (2/3) \cdot \alpha/(\alpha + 1)$.  Then
	$d \geq \min\setd[\big]{(n - k - 3)\cdot \alpha, (m-2)\cdot
	\alpha/(\alpha + 1)}$.
\end{cor}
\begin{proof}
	Note that the proof of \cref{prop:renyi_lb} gave a lower bound on the
	sum $\sum_{y\in \zo^m} P_y^{1+\alpha}$ where $P = \Ext(U_X, U_d)$,
	whereas $d_{\ell_{1+\alpha}}(P, U_m)^{1+\alpha} = \sum_{y\in\zo^m}
	\abs{P_y - 2^{-m}}^{1+\alpha}$.  For $\ell_2$ these can be related
	without any loss, but in general we can use the triangle inequality to
	get
	\begin{align*}
		\D_{1+\alpha}\diver P{U_m}
		&\leq
		\frac{1}{\alpha} \cdot \log \of{2^{m\alpha}
			\cdot \of{d_{\ell_{1+\alpha}}(P, U_m)
		+ 2^{-m\alpha/(\alpha+1)}}^{1+\alpha}}
	\end{align*}
	so that if $d_{\ell_{1+\alpha}}(P, U_m) \leq  \eps_{\alpha} \cdot
	2^{-m\alpha/(1+\alpha)}$ where $\eps_\alpha = (2/3) \cdot
	\alpha/(\alpha + 1) \leq 2^{0.99 \cdot \alpha/(\alpha+1)} - 1$, then
	$\D_{1+\alpha}\diver P{U_m} \leq 0.99$, and we conclude by
	\cref{lem:divergence_bounds,prop:renyi_lb}.
\end{proof}

\subsection{Non-explicit construction}
\label{sec:nonexplicit}

In this section, we show non-constructively the existence of
KL-extractors matching the lower-bound of \cref{prop:small_renyi_lb}
and in particular implying the optimal parameters of standard extractors
for total variation distance. Formally, we will prove:

\begin{thm}\label{thm:kl_nonexplicit}
	For every $n\in\mathbb{N}$, $k\leq n$, and $1>\eps>0$ there is an
	average-case (respectively strong average-case) $(k,\eps)$
	KL-extractor $\Ext:\zo^n \times \zo^d \to \zo^m$ with seed length
	$d = \log(n - k + 1) + \log(1/\eps) + O(1)$ and output length $m = k +
	d - \log(1/\eps) + O(1)$ (respectively $m = k - \log(1/\eps) - O(1)$).
\end{thm}
\begin{rmrk}\label{rmrk:epsgg1}
	For $\eps \gg 1$ the above parameters are not necessarily optimal,
	and it would be interested to get matching upper and lower bounds
	in this regime of parameters.
\end{rmrk}

We will prove \cref{thm:kl_nonexplicit} using the probabilistic method,
analogously to Zuckerman \cite{zuckerman_randomnessoptimal_1997} or
Radhakrishnan and Ta-Shma \cite{rad_ta-_bounds_2000} for total variation
extractors. However, rather than using Hoeffding's inequality, we use
the following lemma:
\begin{lem}\label{lem:random_ext_kl_bound}
	Let $X$ be uniform over a subset of $\zo^n$ of size $K$. Then if
	$\Ext:\zo^n \times \zo^d \to \zo^m$ is a random function, it holds
	for every $\eps > 0$ that
	\[
		\PR[\Ext]{\E[s\sim U_d]{\KL\diver{\Ext(X, s)}{U_m}} > \eps} \leq
		2^{MD-KD\eps/3}
	\]
	where $D = 2^d$ and $M = 2^m$.
\end{lem}
\begin{rmrk}
	For total variation extractors, the analogous bound is
	\[
		\PR[\Ext]{\tv\of[\big]{\of{U_d, \Ext(X, U_d)}, \of{U_d, U_m}} >
		\eps} \leq 2^{MD-2KD\eps^2/\ln 2}.
	\]
	One sees that the bounds are very similar, except the KL divergence
	version depends on $\eps$ rather than $\eps^2$. For the regime where
	$\eps < 1$ the linear dependence is preferable, and is responsible for
	the $1\cdot \log(1/\eps)$ seed length for KL-extractors compared to
	the $2\cdot \log(1/\eps)$ seed length for total variation extractors.
\end{rmrk}
\begin{proof}[Proof of \cref{lem:random_ext_kl_bound}]
	Note that for each $s\in \zo^d$ and fixed $\Ext$, the random variable
	$\Ext(X, s)$ is uniform over the multiset $\setof{\Ext(x,s)}{x\in
	\Supp(X)}$. Hence, since $\Ext$ is a random function, this multiset is
	distributed exactly as taking $K$ iid uniform samples from $\zo^m$, so
	we wish to bound the KL divergence between this empirical distribution
	and the true distribution. For this, the author
	\cite{agrawal_concentration_2019} gave the moment generating function
	bound
	\[
		\E[\Ext]{2^{t\cdot {\KL\diver{\Ext(X, s)}{U_m}}}}
		\leq
		\of{\frac{2^{t/K}}{1-t/K}}^{M-1}
	\]
	for every $0 \leq t < K$, which for $t = K/3$ is at most $2^M$.
	Then since $\Ext(X, s)$ is independent across $s\in \zo^d$, we
	have
	\begin{align*}
		\PR[\Ext]{\E[s\sim U_d]{\KL\diver{\Ext(X, s)}{U_m}} > \eps}
		&=
		\PR[\Ext]{2^{K/3\cdot \sum_{s\in \zo^d}{\KL\diver{\Ext(X, s)}{U_m}}} > 2^{K/3\cdot D\eps}}
		\\
		&\leq 2^{-KD\eps/3}\cdot \prod_{i=1}^{D}2^M\tag*{\qedhere}
	\end{align*}
\end{proof}
We can now prove \cref{thm:kl_nonexplicit}:
\begin{proof}[Proof of \cref{thm:kl_nonexplicit}.]
	We will show that a random function $\Ext:\zo^n\times\zo^d \to \zo^m$
	is a strong average-case $(k,\eps)$ KL-extractor with positive
	probability, the non-strong version then follows from
	\cref{lem:kl_strong_prepend}.  By \cref{lem:graceful_ext_is_average},
	it is enough to prove that $\Ext$ is a strong $(k-t, 2^{t+1}/3\cdot
	\eps)$ KL-extractor for every $t \geq 0$. To reduce the range of
	$t$ we need to consider, note that it suffices to be a
	$(\log\floor{2^{k-t}}, 2^{t+1}/3\cdot \eps)$ extractor for every
	$t\geq 0$, so that by rounding down it is enough to be a $(k - t,
	2^t/3\cdot \eps)$ strong KL-extractor for each $t\geq 0$ such that
	$2^{k-t}$ is an integer.

	Now, consider a fixed $t\geq 0$ such that $2^{k-t}$ is an integer.
	Since the KL divergence is convex in its first argument and all
	distributions of min-entropy at least $k-t$ are convex combinations of
	``flat'' distributions which are uniform over a set of size $2^{k-t}$
	(Chor and Goldreich \cite{cho_gol_unbiased_1988}), it suffices to
	analyze the behavior of $\Ext$ on such distributions. Then for every
	subset $X\subseteq \zo^n$ of size $2^{k-t}$,
	\cref{lem:random_ext_kl_bound} tells us that
	\[
		\PR[\Ext]{\E[s\sim U_d]{\KL\diver{\Ext(U_X, s)}{U_m}} > 2^t/3\cdot \eps} \leq
		2^{MD-2^{k-t}\cdot D\cdot (2^t/3\cdot \eps)/3}
		=
		2^{MD-KD\eps/9}
	\]
	where $M = 2^m$, $D = 2^d$, and $K = 2^k$. There are $\sum_{j=0}^K
	\binom{N}{j}$ such subsets $X$ of $\zo^n$ for which we simultaneously need to
	establish that $\E[s\sim U_d]{\KL\diver{\Ext(U_X, s)}{U_m}} \leq
	2^t/3\cdot \eps$, so we have by a union bound that the probability that
	$\Ext$ is not a strong average-case $(k, \eps)$ KL-extractor is at most
	\begin{align*}
		2^{MD-KD\eps/9}\cdot \sum_{j=0}^K \binom{N}{j}
		&\leq 2^{MD-KD\eps/9}\cdot \of{\frac{Ne}{K}}^K
		= 2^{MD+K\log(Ne/K)-KD\eps/9}.
	\end{align*}
	Hence, as long as
	\begin{align*}
		MD &< \frac{KD\eps}{18} & K\log\of{\frac{Ne}{K}}&< \frac{KD\eps}{18}\\
		m &\leq k - \log(1/\eps) - O(1) & d &\geq \log(n-k+1) + \log(1/\eps) + O(1)
	\end{align*}
	we know that a random function is a strong average-case $(k, \eps)$
	KL-extractor with positive probability as desired.
\end{proof}

\section{Constructions of subgaussian samplers}
\label{sec:subgaussian_samplers}
\subsection{Subconstant \titleeps/ and \titledelta/}

The goal of this section is to establish the following theorem, which is
our explicit construction of subgaussian samplers with sample complexity
having no dependence on $m$, and with randomness complexity and sample
complexity matching the best-known $[0,1]$-valued sampler when $\eps$
and $\delta$ are subconstant (up to the hidden polynomial in the sample
complexity).
\begin{thm}\label{thm:subgaussian_sampler}
	For all $m\in \mathbb{N}$, $1>\eps,\delta > 0$, and $\alpha > 0$ there
	exists an explicit $(\delta, \eps)$ absolute averaging sampler
	(respectively strong absolute averaging sampler) for subgaussian and
	subexponential functions $\Samp:\zo^n \to \of{\zo^m}^D$ with sample
	complexity $D = \poly(\log(1/\delta), 1/\eps)$ and randomness
	complexity $n = m + (1+\alpha)\cdot \log(1/\delta)$ (respectively $n =
	m + (1+\alpha)\cdot\log(1/\delta) + 2\log(1/\eps) +
	O(1)$).
\end{thm}

We will use essentially the same construction used for bounded samplers
in this regime, namely applying the Reingold, Wigderson, and Vadhan
\cite{rei_vad_wig_entropy_conference} zig-zag product for extractors to
combine the expander extractor of Goldreich and Wigderson
\cite{gol_wig_tiny_1997} and an extractor with logarithmic seed length.
However, as described in detail in \cref{sec:composition}, even the
basic composition used in this construction does not work for general
subgaussian extractors, so we will instead use the zig-zag product for
KL-extractors (\cref{cor:zigzag}) combining extractors for R\'enyi
divergences, specifically the $\D_2$-extractor from
\cref{cor:expander_extractor} and the KL-extractor from
\cref{cor:kl_guv_lessloss}, to get the following high-entropy
KL-extractor:
\begin{thm}\label{thm:high_entropy_kl_extractor}
	For all integers $m$ and $1>\alpha,\delta,\eps > 0$ there is an
	explicit average-case (respectively strong average-case) $(k, \eps)$
	KL-extractor $\Ext:\zo^n \times \zo^d \to \zo^m$ with $n = m +
	(1+\alpha)\log(1/\delta) - O(1)$ (respectively $n = m + (1 +
	\alpha)\cdot\log(1/\delta) + \log(1/\eps) + O(1)$), $k = n
	- \log(1/\delta)$, and $d = O_{\alpha}(\log(\log(1/\delta)/\eps))$.
\end{thm}
\begin{proof}
	We prove the claim for strong extractors, for the non-strong claim one
	can simply define $\Ext(x, (s, t)) = \Ext_{strong}((x, t), s)$ where
	$t$ has length $\log(1/\eps) + O(1)$.

	By \cref{cor:expander_extractor}, there is a universal constant $C >
	0$ such that for $d_{out} = \ceil{C\log(1/(\delta\eps))} \leq
	C\log(1/\delta) + C\log(1/\eps) + 1$ there is an explicit average-case
	$(n_{out} - \log(1/\delta), \eps/4)$ $\D_2$-extractor
	$\Ext_{out}:\zo^{n_{out}}\times \zo^{d_{out}} \to \zo^{n_{out}}$ with
	$n_{out} = m - d_{out}$. Furthermore, $\Ext_{out}$ has the property
	that the function $\Waste_{out}(x, s) = s$ is such that $(\Ext_{out},
	\Waste_{out})$ is an injection.

	Let $k_{in}' = C\log(1/\delta)/(1-\beta)$, $k_{in}'' =
	(C+1)\log(1/\eps) + O(1)$, and $k_{in} = k_{in}' + k_{in}''$ for
	$0<\beta<1$ some parameter to be chosen later. Then by
	\cref{cor:kl_guv_lessloss}, there is an explicit $(k_{in}, \eps/4)$
	strong average-case KL-extractor $\Ext_{in}:\zo^{n_{in}} \times
	\zo^{d_{in}} \to \zo^{m_{in}}$ with $n_{in} = k_{in} +
	\log(1/\delta)$, $d_{in} = O_{\beta}(\log(n_{in}/\eps)) + O(k_{in}'')
	= O_{\beta}(\log(\log(1/\delta)/\eps))$, and $m_{in} = (1-\beta)
	k_{in}' + k_{in}'' -\log(1/\eps) - O(1) = d_{out}$.  Furthermore, the
	function $\Waste_{in}(x, s) = (x, s)$ is an injection.

	Furthermore, for $k_{waste} = (n_{out} + n_{in} - \log(1/\delta)) -
	n_{out} = n_{in} - \log(1/\delta) = k_{in}=k_{in}'+k_{in}''$, by
	\cref{cor:kl_guv_lessloss} there is also an explicit $(k_{waste},
	\eps/4)$ strong average-case KL-extractor
	$\Ext_{waste}:\zo^{d_{out}+n_{in}+d_{in}} \times \zo^{d_{waste}} \to
	\zo^{m_{waste}}$ such that $m_{waste} = d_{out}$ and $d_{waste} =
	O_{\beta}\of{\log\of{(d_{out}+n_{in}+d_{in})/\eps}} + O(k_{in}'') =
	O_{\beta}\of{\log(\log(1/\delta)/\eps)}$.

	Then by the zig-zag product for KL-extractors (\cref{cor:zigzag}),
	there is an explicit $(n_{out} + n_{in} - \log(1/\delta), \eps)$
	strong average-case KL-extractor $\Ext:\zo^{n_{out}+n_{in}} \times
	\zo^{d_{in}+d_{waste}} \to \zo^{n_{out} + m_{waste}}$, where we have
	\begin{align*}
		n_{out} + n_{in}
		&= \of{m - d_{out}} + \of[\Big]{\of[\big]{C\log(1/\delta)/(1-\beta) + (C+1)\log(1/\eps)
		+ O(1)} + \log(1/\delta)}\\
		&\leq m + \log(1/\delta) + \log(1/\eps) + \log(1/\delta)\cdot C\cdot \of[\big]{1/(1-\beta)-1} + O(1)\\
		d_{in} + d_{waste} &= O_{\beta}\of{\log(\log(1/\delta)/\eps)}\\
		n_{out} + m_{waste} &= \of{m-d_{out}} + d_{out} = m.
	\end{align*}
	Choosing $\beta = \alpha/(\alpha + C)$ so that $C\cdot
	\of[\big]{1/(1-\beta)-1}\leq \alpha $ gives the claim.
\end{proof}

We can now prove \cref{thm:subgaussian_sampler}.
\begin{proof}[Proof of \cref{thm:subgaussian_sampler}]
	Let $\Ext:\zo^n \times \zo^d \to \zo^m$ be the explicit
	$\of{n-\log\of[\big]{1/(\delta/2)}, \eps^2}$ KL-extractor (respectively strong
	KL-extractor) of \cref{thm:high_entropy_kl_extractor}, so that $d =
	O_{\alpha}\of{\log\log(1/\delta)/\eps}$ and $n = m + (1 +
	\alpha)\log(1/\delta)$ (respectively $n = m + (1 +
	\alpha)\log(1/\delta) + 2\log(1/\eps) + O(1)$).

	Then by \cref{lem:dg_vs_kl,lem:divergence_bounds}, $\Ext$ is also an
	$\of{n-\log\of[\big]{1/(\delta/2)}, \eps}$ $d_{\cE}$-extractor
	(respectively strong $d_{\cE}$-extractor), so by
	\cref{thm:func_ext_are_samplers} the function $\Samp:\zo^n \times
	\of{\zo^m}^D$ given by $\Samp(x)_i = \Ext(x, i)$ is an explicit
	$(\delta/2 ,\eps)$ sampler for $\cE$ (respectively strong sampler for
	$\cE$), and thus by symmetry of $\cE$ an explicit $(\delta, \eps)$
	absolute subexponential sampler (respectively absolute strong
	subexponential sampler) as desired.
\end{proof}

\subsection{Constant \titledelta/}
We recall from the introduction that the pairwise independent sampler of
Chor and Goldreich \cite{cho_gol_power_1989} works for subgaussian
functions, and in fact the more general class of functions with bounded
variance. The sampler has exponentially worse dependence on $\delta$
than is necessary for subgaussian samplers, but is very simple and has
randomness complexity optimal up to constant factors.
\begin{thm}[\cite{cho_gol_power_1989}]\label{thm:pi_samp}
	For all $m\in \mathbb{N}$ and $1>\eps, \delta > 0$ with $1/(\delta\eps^2) < 2^m$,
	there is an explicit strong sampler $\Samp:\zo^n \to \of{\zo^m}^{D}$
	for functions with bounded variance $\mathcal{M}_2$, with randomness
	complexity $n = O(m)$ and sample complexity $D =
	O\of{\frac{1}{\eps^2\delta}}$ defined as $\Samp(h)_d = h(d)$ where $h$
	is drawn at random from a size $2^n$ pairwise-independent hash family
	$\mathcal{H}$ of functions from $[D] \to \zo^m$.
\end{thm}
\begin{proof}
	The fact that pairwise independence gives rise to a strong
	bounded-variance sampler is immediate by Chebyshev's inequality. The
	existence of a pairwise indepenent hash family with the claimed
	parameters is due to Chor and Goldreich \cite{cho_gol_power_1989},
	with similar constructions in the probability literature due to Joffe
	\cite{joffe_sequence_1971}.
\end{proof}

We also show that the Expander Neighbor sampler of
\cite{kar_pip_sip_timerandomness_1985,gol_wig_tiny_1997} is a
bounded-variance sampler.
\begin{thm}\label{thm:ramanujan_sampler}
	There is a universal constant $C \geq 1$ such that for all $m\in
	\mathbb{N}$ and $1>\eps, \delta > 0$ there is an explicit
	sampler $\Samp:\zo^n \to \of{\zo^m}^{D}$ for functions with bounded
	variance $\mathcal{M}_2$, with randomness complexity $n = m$ and
	sample complexity $D = O\of{\of{\frac{1}{\eps^2\delta}}^C}$. Moreover, if
	the algorithm is given access to a consistently labelled neighbor
	function of a Ramanujan graph over $\zo^n$ of
	degree $O(1/(\delta\eps^2))$, then one can take $C = 1$.
\end{thm}
\begin{proof}
	By \cref{cor:expander_extractor}, there is an explicit $(n -
	\log(1/\delta), \eps\cdot 2^{-m/2})$ $\ell_2$-extractor $\Ext:\zo^m
	\times \zo^d \to \zo^m$ with $d = \ceil{C(\log(1/\delta) +
	2\log(1/\eps))} + O(1)$, where one can take $C = 1$ given the assumed
	Ramanujan graph. Then by \cref{thm:lp_is_bounded_q_moments} $\Ext$ is
	also an $(n-\log(1/\delta), \eps)$ $\mathcal{M}_2$-extractor, so we
	conclude by \cref{thm:func_ext_are_samplers}.
\end{proof}
\begin{rmrk}
Note that given explicit constructions of Ramanujan graphs,
\cref{thm:ramanujan_sampler} has the same sample complexity but better
randomness complexity than the sampler of \cref{thm:pi_samp}.
\end{rmrk}

\subsection{Non-explicit construction}\label{sec:nonexplicit_samp}

Applying \cref{lem:divergence_bounds,lem:dg_vs_kl} to
\cref{thm:kl_nonexplicit} our non-explicit construction of KL-extractors
gives:
\begin{cor}\label{cor:subgaussian_nonexplicit}
	For every $n\in\mathbb{N}$, $k\leq n$, and $1>\eps>0$ there is an
	average-case (respectively strong average-case) $(k,\eps)$
	$d_{\cE}$-extractor $\Ext:\zo^n\times\zo^d\to\zo^m$ with $d= \log(n-k+1) +
	2\log(1/\eps) + O(1)$ and $m \geq k + d - 2\log(1/\eps) - O(1)$
	(respectively $m \geq k - 2\log(1/\eps) - O(1)$)
\end{cor}
Since $d_{\cE}$-extractors are also total variation extractors,
\cref{cor:subgaussian_nonexplicit} is optimal up to additive constants
by the lower bound of Radhakrishnan and Ta-Shma
\cite{rad_ta-_bounds_2000}.

Using the fact that extractors are samplers
(\cref{thm:func_ext_are_samplers}), we get
\begin{cor}\label{cor:subgaussian_samp_nonexplicit}
	For every integer $m$ and $1>\delta,\eps > 0$ there is a
	$(\delta,\eps)$ sampler (respectively strong sampler) $\Samp:\zo^n \to
	\of{\zo^m}^D$ for subgaussian and subexponential functions with sample
	complexity $D = O\of{\frac{\log1/\delta}{\eps^2}}$ and randomness
	complexity $n = m + \log(1/\delta) - \log\log(1/\delta) + O(1)$
	(respectively $n =  m + \log(1/\delta) + 2\log(1/\eps) + O(1)$).
\end{cor}
Note that this matches the best-known (non-explicit) parameters of
averaging samplers for $[0,1]$-valued functions due to Zuckerman
\cite{zuckerman_randomnessoptimal_1997}.

\section{Acknowledgements}
The author would like to thank Jaros\l{}aw B\l{}asiok for suggesting the
problem of constructing subgaussian samplers and for helpful discussions
and feedback, Salil Vadhan for many helpful discussions and his
detailed feedback on this writeup, and the anonymous reviewers for their
helpful comments and feedback.

\bibliographystyle{IEEEtranSA}
\bibliography{main,misc}

\end{document}